\documentclass[final]{dmtcs-episciences}

\usepackage{thm-restate}

\usepackage[utf8]{inputenc}
\usepackage{subfigure}

\usepackage[round,numbers]{natbib}

\author[Ma{\"e}l Dumas et al.]{Ma{\"e}l Dumas\affiliationmark{1}\thanks{A majority of the work for this paper was conducted when M. Dumas was affiliated with Universit{\'e} d'Orl{\'e}ans.}
  \and Anthony Perez\affiliationmark{2}
  \and Mathis Rocton\affiliationmark{3}
  \and Ioan Todinca\affiliationmark{2}
}
\title[Polynomial kernels for block and strictly chordal editing]{Polynomial kernels for edge modification problems towards block and strictly chordal graphs\thanks{A preliminary version of this work appeared in IPEC'21~\cite{DPT21}. All results concerning \BGE{} are new to this submission.}}
\affiliation{
    Institute of Informatics, University of Warsaw, Warsaw, Poland\\
    Université d’Orléans, INSA CVL, LIFO UR 4022, Orléans, France\\
    Algorithms and Complexity Group, TU Wien, Vienna, Austria
  }

\usepackage{enumerate}
\def\ie{{\em i.e.}~}

\newcommand{\C}{\mathcal{C}}

\renewcommand{\leq}{\leqslant}
\renewcommand{\geq}{\geqslant}

\def\SCC{{\textsc{Strictly Chordal Completion}}}
\def\SCD{{\textsc{Strictly Chordal Deletion}}}
\def\SCE{{\textsc{Strictly Chordal Editing}}}

\def\BGC{{\textsc{Block Graph Completion}}}
\def\BGD{{\textsc{Block Graph Deletion}}}
\def\BGE{{\textsc{Block Graph Editing}}}
\def\CD{{\textsc{Cluster Deletion}}}
\def\CE{{\textsc{Cluster Editing}}}
\def\LPC{{\textsc{3-Leaf Power Completion}}}

\newcommand{\bdsc}{strictly chordal}
\newcommand{\bg}{block}

\usepackage[svgnames,x11names]{xcolor}

\usepackage[utf8]{inputenc}
 
\usepackage{fixmath}
\usepackage{stmaryrd}

\usepackage{amsmath, amssymb}
\newtheorem{theorem}{Theorem}
\newtheorem{lemma}{Lemma}

\newtheorem{definition}{Definition}

\newtheorem{conjecture}{Conjecture}
\newtheorem{observation}{Observation}

\newtheorem{polyrule}{Rule}[section]

\usepackage[capitalise,nameinlink]{cleveref}
\crefname{lemma}{Lemma}{Lemmata}
\crefname{figure}{Figure}{Figures}
\crefname{polyrule}{Rule}{Rules}
\crefname{observation}{Observation}{Observations}
\crefname{conjecture}{Conjecture}{Conjectures}

\usepackage[linesnumbered,vlined,ruled]{algorithm2e}

\usepackage[defaultlines=3,all]{nowidow}

\usepackage[skins]{tcolorbox}
\tcbuselibrary{many}

\newtcolorbox{mypb}[2][]
{
    enhanced,
    boxed title style = {colframe=white},
    attach boxed title to top left={
        xshift=0.5cm,
        yshift= -3.5mm,     
    },
    top=4mm,
    coltitle=black,
    beforeafter skip=\baselineskip,
    colframe = lightgray,
    colback  = white,
    colbacktitle  = white,
    coltitle = black,  
    fonttitle = \scshape,
    titlerule = 0mm, 
    title    = {#2},
    #1
}

\newcommand{\Pb}[4]{%
    \begin{mypb}{#1}
       \textbf{\textsf{Input}}: #2%
       \par\noindent%
       \textbf{\textsf{#4}}: Does there exist #3?%
       \smallskip%
       \par\noindent%
    \end{mypb}
}

\usepackage{tikz}
\usetikzlibrary{positioning}
\usetikzlibrary{calc}
\tikzstyle{fptbox}=[line width=0.5mm,rectangle, minimum height=.8cm,fill=white!70,rounded corners=1mm,draw]
\tikzstyle{us}=[line width=1mm,rectangle, minimum height=.8cm,fill=white!70,rounded corners=1mm,draw]
\tikzstyle{fptedge}=[line width=0.5mm,->]

\usepackage[framemethod=tikz]{mdframed}

\definecolor{darkred}{RGB}{170,10,50}
\definecolor{normalred}{RGB}{220,0,0}
\definecolor{grassgreen}{RGB}{100,160,20}
\definecolor{darkgreen}{RGB}{0,150,20}
\definecolor{pacificorange}{RGB}{220,100,0}

\begin{document}
\publicationdata{vol. 27:2}{2025}{5}{10.46298/dmtcs.12998}{2024-02-05; 2024-02-05; 2024-11-12}{2025-02-14}

\maketitle

\begin{abstract}
We consider edge modification problems towards block and strictly chordal graphs, where one is given an undirected graph $G = (V,E)$ and an integer $k \in \mathbb{N}$ and seeks to \emph{edit} (add or delete) at most $k$ edges from $G$ to obtain a block graph or a strictly chordal graph. The \emph{completion} and \emph{deletion} variants of these problems are defined similarly by only allowing edge additions for the former and only edge deletions for the latter. 
Block graphs are a well-studied class of graphs and admit several characterizations, \emph{e.g.} they are diamond-free chordal 
graphs. Strictly chordal graphs, also referred to as \emph{block duplicate graphs}, are a natural generalization of block graphs where one can add true twins of cut-vertices. 
Strictly chordal graphs are exactly dart and gem-free chordal graphs.
We prove the NP-completeness for most variants of these problems and provide:
\begin{itemize}
    \item $O(k^2)$ vertex-kernels for \textsc{Block Graph Editing} and \textsc{Block Graph Deletion}
    \item $O(k^3)$ vertex-kernels for \textsc{Strictly Chordal Completion} and \textsc{Strictly Chordal Deletion}
    \item an $O(k^4)$ vertex-kernel for \textsc{Strictly Chordal Editing}
\end{itemize} 
\end{abstract}

\section{Introduction}
Parameterized algorithms are among the most natural approaches to tackle NP-hard optimization problems~\cite{CFK+15}. 
 In particular, they have been very successful in dealing with so-called edge modification problems on graphs: given as input an arbitrary graph $G=(V,E)$ and a parameter $k \in \mathbb{N}$, the goal is to transform $G$ into a graph with some specific properties (i.e., belonging to a specific graph class $\mathcal{G}$) by adding and/or deleting at most $k$ edges. Parameterized algorithms (also called FPT for \emph{fixed parameter tractable}) aim at a time complexity of type $f(k)\cdot n^{O(1)}$, where $f$ is some computable function, hence the combinatorial explosion is restricted to parameter $k$. 

When the target class $\mathcal{G}$ is characterized by a finite family of forbidden induced subgraphs, modification problems are FPT by a result of Cai~\cite{Cai96}. Indeed, as long as the graph contains one of the forbidden subgraphs, one can try each possibility to correct this obstruction and branch by recursive calls. On each branch, the budget $k$ is strictly diminished, therefore the whole algorithm has a number of calls bounded by some function $f(k)$. The situation is more complicated when the target class $\mathcal{G}$ is characterized by an infinite family of forbidden induced subgraphs. Nonetheless, a large literature is devoted to edge modification problems towards chordal graphs (where we forbid all induced cycles with at least four vertices) as well as sub-classes of chordal graphs, typically obtained by requiring some fixed set of obstructions, besides the long cycles. Observe that, in this case, the situation remains relatively simple if we restrict ourselves to \emph{edge completion} problems, where we are only allowed to \emph{add} edges to the input graphs. Indeed, in this case, if a graph has an induced cycle of length longer than $k+3$, it cannot be made chordal by adding at most $k$ edges. Therefore we can use again the approach of Cai to deal with cycles of length at most $k+3$ and other obstructions, and either the algorithm finds a solution in $f(k)$ recursive calls, or we can conclude that we have a no-instance. The cases of \emph{edge deletion} problems (where we are only allowed to remove edges) and \emph{edge editing} problems (where we are allowed to both remove edges or add missing edges) are more complicated, since even long cycles can be eliminated by a single edge removal. Therefore more efforts and more sophisticated techniques were necessary in these  situations, but several such problems turned out to be FPT~\cite{CM16,DGH+06,DGHN08}.  
The interested reader can refer to~\cite{CDF+23} for a broad and comprehensive survey on parameterized algorithms for edge modification problems. \\

We focus here on a sub-family of parameterized algorithms, namely on \emph{kernelization}. The goal of kernelization is to provide a polynomial algorithm transforming  any instance $(I, k)$ of the problem into an equivalent instance $(I',k')$ where $k'$ is upper bounded by some function of $k$ (in our case we will simply have $k' \leq k$), and the size of the new instance $I'$ is upper bounded by some function $g(k)$. Hence the size of the reduced instance does not depend on the size of the original instance. Kernels are obtained through a set of \emph{reduction rules}. 
While kernelization is possible for all FPT problems (the two notions are actually equivalent), 
the interesting question is whether a given FPT problem admits \emph{polynomial kernels}, where the size of the reduced instance is bounded by some polynomial in $k$. Note that, under some complexity assumptions, not all FPT problems admit polynomial kernels~\cite{BDF+09,BJK14,BTY11,CC15,GHP+13,KW09}.

In this paper we focus on \BGE, \BGD{} and all three 
variants of modification problems towards strictly chordal graphs. \emph{Block} graphs 
are a well-studied subclass of chordal graphs, that are diamond-free (\cref{fig:obs}). 
They are moreover the graphs in which every biconnected component induces a clique. Notice that this observation allows for a simple polynomial-time algorithm for \BGC{} since one needs to turn every biconnected component of the input graph into a clique.
\emph{Strictly chordal} graphs are another subclass of chordal graphs, also known as \emph{block duplicate} graphs~\cite{Kennedy05,KLY06,GP02}. They can be obtained from block graphs by repeatedly choosing some cut-vertex $u$ and adding a \emph{true twin} $v$ of $u$, that is a vertex $v$ adjacent to $u$ and all neighbors of $u$~\cite{GP02}. They can also be characterized as dart, gem-free chordal graphs (see \cref{fig:obs} and next section) and are thus ptolemaic (\ie chordal distance-hereditary~\cite{BLS99}). 
Strictly chordal graphs are also known to be a subclass of 4-leaf power graphs~\cite{KLY06}, and a super-class of 3-leaf power graphs~\cite{DGH+06}. 

\paragraph{Related work}
Kernelization for \textsc{chordal completion} goes back to the '90s and the seminal paper of Kaplan, Shamir and Tarjan~\cite{KST99}. In this work, the 
authors provide an $O(k^3)$ vertex-kernel for \textsc{Chordal Completion} (also known as \textsc{Minimum Fill-In}) that was later improved to $O(k^2)$ by a tighter analysis~\cite{NSS00}. Since then, several authors addressed completion, deletion and/or editing problems towards sub-classes of chordal graphs, such as 3-leaf power graphs~\cite{BPP10}, split and threshold graphs~\cite{Guo07}, proper interval graphs~\cite{BP13}, trivially perfect graphs~\cite{BBC+22,DP18,DPT23,Guo07,DP23} or ptolemaic graphs~\cite{CGP21}. We provide a general picture for some of those classes \cref{fig:polykernel_diagramme}.

\begin{figure}[!hbt]
\centering
\scalebox{0.85}{\begin{tikzpicture}[node distance=10mm]
    \begin{scope}
      \node[fptbox] (chordal) at (0,0) {chordal~\cite{KST99,NSS00}};
      
      \node[fptbox] (4leaf) [left = of chordal] {$4$-leaf power} edge[fptedge] (chordal);
      
      \node[fptbox] (it) [right = of chordal, xshift=14mm] {interval} edge[fptedge] (chordal);
      
      \node[fptbox] (pto) [below = of chordal, xshift=-30mm, yshift=5mm] {ptolemaic~\cite{CGP21}} edge[fptedge] (chordal); 
      
      \node[fptbox] (tp) [below = of it, right = of pto, yshift=-15mm] {trivially perfect~\cite{BBC+22,DP23}} edge[fptedge] (pto) edge[fptedge] (it);

      \node[us] (sc) [left = of pto, xshift=4mm, yshift=-15mm] {strictly chordal} edge[fptedge] (pto) edge[fptedge] (4leaf);
      
      \node[fptbox] (split) [right = of tp] {split~\cite{BBC+22,hammer1981splittance}} edge[fptedge] (chordal);
      
      \node[fptbox] (threshold) [below  = of tp, yshift=5mm] {threshold~\cite{DDL+22}} edge[fptedge] (split) edge[fptedge] (tp);
      
      \node[fptbox] (pint) [below = of it, yshift=5mm] {proper interval~\cite{BP13}} edge[fptedge] (it); 
      
      \node[fptbox] (3leaf) [left = of sc, below = of sc, xshift=-5mm, yshift=5mm] {$3$-leaf power~\cite{BPP10}} edge[fptedge] (sc);
      
      \node[us] (block) [right = of 3leaf, below = of sc, xshift=15mm, yshift=5mm] {block} edge[fptedge] (sc);
      
      \draw[color=Coral3,line width=2pt,-,rounded corners] 
      ($(pto.west |- pto.north)+(-4.5,0.3)$) -- 
      ($(chordal.west |- chordal.south)+(-0.6,-0.3)$) -- 
      ($(chordal.north west)+(-0.1,0.2)$) -- 
      ($(chordal.north east)+(0.1,0.2)$) -- 
      ($(chordal.east|- chordal.south)+(0.6,-0.3)$) -- 
      ($(pint.north east)+(0.2,0.2)$);

      \draw[color=DeepSkyBlue4,line width=2pt,-,rounded corners] 
      ($(pto.west |- pto.south)+(-4.5,-0.3)$) -- 
      ($(pint.east |- pint.south)+(0.2,-0.3)$);

      
      
      
      \node (pk)[xshift=5mm] at ($(4leaf.north west)+(-3.55,-1.4)$) {\textbf{poly-kernel $\downarrow$ }};

      \node (comp) at ($(4leaf.north west)+(-3.2,-0.8)$) {{\color{Coral3}\textbf{Completion}}};
      \node (ed) at ($(4leaf.south west)+(-2.7,-1.4)$) {{\color{DeepSkyBlue4}\textbf{Editing, Deletion}}};
      
    \end{scope}
\end{tikzpicture}}
     \caption{Kernelization status of subclasses of chordal graphs. 
     A class below the red or blue line indicates that the corresponding completion, deletion or editing problem admits a polynomial kernel.  All the edge modification problems for the presented classes are NP-complete at the exception of the ones for $4$-leaf power for which the complexity is unknown, \textsc{Split Editing} which is surprisingly in P~\cite{hammer1981splittance}, and \BGC{} which is also in $P$ as observed previously.
    }
     \label{fig:polykernel_diagramme}
\end{figure}

All these classes have in common that they can be defined as chordal graphs, plus a constant number of obstructions. 
 Moreover, 
several known results deal with classes that share a very common feature (see e.g.~\cite{BP13,KU14,BPP10,DP23}).  
Very informally, the target class $\mathcal G$ admits a tree-like decomposition, in the sense that the vertices of any graph $H \in {\mathcal G}$ can be partitioned into clique modules (having the same neighborhood in $H$), and these modules can be mapped onto the nodes of a decomposition tree, the structure of the tree describing the adjacencies between modules. Therefore, if an arbitrary graph $G$ can be transformed into graph $H$ by at most $k$ edge additions or deletions, at most $2k$ modules can be affected by the modifications. By removing the affected nodes from the decomposition tree, we are left with several components (\emph{chunks}) that correspond, in the initial graph $G$ as well as in $H$, to induced subgraphs that may be large but that already belong to the target class. Moreover, these chunks are attached to the rest of graph $G$ in a very regular way, through one or two nodes of the decomposition tree. The kernelization algorithms need to analyze these chunks and provide reduction rules, typically by ensuring a small number of nodes in the decomposition tree, plus the fact that each node corresponds to a module of small size. %
However, finding a general framework that captures this particular decomposition does not seem to be an easy task as each class exhibits particular structural properties that may need to be dealt with independently. 

As mentioned earlier, when restricted to completion problems any input graph cannot have large induced cycles and this property may be exploited to design kernelization algorithms. This observation lead to the following conjecture by Bessy and Perez~\cite{BP13}. 
\begin{conjecture}[\cite{BP13}]
\label{conj:bp}
    Let $\mathcal{G}$ be a class of graphs that can be defined as chordal graphs plus a constant number of obstructions
    Then, the \textsc{$\mathcal{G}$-Completion} problem admits a polynomial kernel. 
\end{conjecture}
Several other questions remain open, for example it is not known whether \textsc{chordal deletion} or \textsc{chordal editing} admit polynomial kernels~\cite{CDF+23}. 


\paragraph{Our contribution}
We first prove that most variants of the considered problems are NP-Complete, the 
only exception being the \BGC{} problem that admits a polynomial-time algorithm. 
Secondly, we give kernelization algorithms for all these problems, with different 
bounds. 
First of all, we begin by illustrating the techniques we will use on the \BGE{} problem and obtain the following result.

\begin{restatable}{theorem}{THMBGE} 
\label{thm:taille_noyau_BGE}
\BGE{} and \BGD{} admit a kernel with $O(k^2)$ vertices.
\end{restatable}

Next, we present kernelization algorithms for the \SCE{} problem and obtain better vertex-kernels for both \SCC{} and \SCD{}.

\begin{restatable}{theorem}{THMSCE}
    \label{thm:taille_noyau_SCE}
\SCE{} admits a kernel with $O(k^4)$ vertices.
\end{restatable}

\begin{restatable}{theorem}{THMSCCD}
    \label{thm:taille_noyau_SCCD}
    \SCC{} and \SCD{} admit kernels with $O(k^3)$ vertices.
\end{restatable}

Above all, our purpose is to exhibit general techniques that might, we hope, be extended to kernelizations for edge modification problems towards other graph classes. In particular, unlike the general strategy we previously described, both the classes of \bg\ graphs and of \bdsc\ graphs do not have exactly a tree-like decomposition. Still, they can be decomposed into structures than can be seen as a generalization of a tree. Our algorithms exploit these informal observations and provide the necessary reduction rules together with the combinatorial analysis for the kernel size. 
 Finally, we note that while both kernels use structures and rules that are similar, the ones for \BGE{} are significantly simpler, hence the kernel for this problem is a nice warm-up for the one for \SCE{}. Moreover, even though reduction rules for \SCE{} may be directly adapted for \BGE, it would lead to a significantly larger kernel.  



\paragraph{Outline} The paper is organized as follows. We begin with some preliminary results and definitions, including the 
NP-completeness proofs of aforementioned problems (\cref{sec:prelim}). 
We then describe a kernel with 
$O(k^2)$ vertices for \BGE{} and \BGD{} (\cref{sec:block}), introducing the techniques that will be adapted to obtain an $O(k^4)$ vertex-kernel for \SCE{} (\cref{sec:sc}). 
We then explain how these techniques produce $O(k^3)$ vertex-kernels for both 
\SCC{} and \SCD{} (\cref{sec:sccd}) and then conclude with some perspective 
and open problems (\cref{sec:conclu}). 

\section{Preliminaries} 
\label{sec:prelim}

We consider simple, undirected graphs $G = (V,E)$ where $V$ denotes the \emph{vertex set} and $E \subseteq (V \times V)$ the %
\emph{edge set} of $G$. We will sometimes use $V(G)$ and $E(G)$ to clarify the context. 
Given a vertex $u \in V$, the \emph{open neighborhood} of $u$ is the set %
$N_G(u) = \{v \in V:\ uv \in E\}$. 
The \emph{closed neighborhood} of $u$ is defined as $N_G[u] = N_G(u) \cup \{u\}$. 
Two vertices $u$ and $v$ are \emph{true twins} if $N_G[u] = N_G[v]$. 
Given a subset of vertices $S \subseteq V$, $N_G[S]$\ is the set $\cup_{v \in S} N_G[v] $
and $N_G(S)$ is the set $N_G[S] \setminus S$. 
We will omit the mention to $G$ whenever the context is clear. A subset of vertices $M \subseteq V$ is a \emph{module} if for every vertices $x,y \in M$, $N(x) \setminus M = N(y) \setminus M$. 
The subgraph \emph{induced} by $S$ is defined as %
$G[S] = (S,E_S)$ where $E_S = \{uv \in E:\ u \in S, v \in S\}$. 
For the sake of readability, given a subset $S \subseteq V$ 
we define $G\setminus S$ as $G[V \setminus S]$.
A subgraph $C$ is a \emph{connected component} of $G$ if it is a maximal connected subgraph of $G$. 
A graph is \emph{biconnected} if it is still connected after removing any vertex.
A subgraph $C$ is a \emph{biconnected component} of $G$ if it is a maximal biconnected subgraph of $G$. 
A set $S \subseteq V$ is a \emph{separator} of $G$ %
if $G\setminus S$ is not connected. Given two vertices $u$ and $v$ of $G$, the separator $S$ is a $uv$-separator if $u$ and $v$ lie in distinct connected components of $G \setminus S$. Moreover, $S$ is a \emph{minimal} $uv$-separator if no proper subset 
of $S$ is a $uv$-separator. Finally, a separator $S$ is \emph{minimal} if there exists 
a pair $\{u,v\}$ such that $S$ is a minimal $uv$-separator.

\begin{figure}[ht]
    \centering
    \includegraphics[width=11cm]{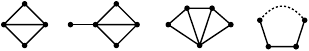}
    \caption{The \emph{diamond}, \emph{dart}, \emph{gem} and \emph{hole} (cycle of length at least $4$)}
    \label{fig:obs}
\end{figure}

\begin{definition}
\label{def:cc}
    Given a graph $G = (V,E)$, a \emph{critical clique} of $G$ is a set $K \subseteq V$ such that $G[K]$ is a clique, $K$ is a module and is inclusion-wise maximal under this property.
\end{definition}

Notice that $K$ is a maximal set of true twins and that the set $\mathcal{K}(G)$ of critical cliques of any graph $G$ partitions its vertex set $V(G)$.  This leads to the following definition. 

\begin{definition}[Critical clique graph]
\label{def:ccg}
Let $G = (V,E)$ be a graph. The \emph{critical clique graph} of $G$ is the graph $\C(G) = (\mathcal{K}(G), E_\mathcal{C})$ with $E_\mathcal{C} = \{KK'\;|\;\forall u\in K, \forall v\in K', uv \in E\}$.
\end{definition}

\paragraph{Block and strictly chordal graphs}
Block graphs are graphs in which every biconnected component is a clique. They can also be characterized as chordal graphs that do not contain diamonds as induced subgraph~\cite{BM86} (see \cref{fig:obs}).
A natural generalization of block graphs are \bdsc\ graphs, also known as block 
duplicate graphs, that are obtained from block graphs by adding true twins~\cite{Kennedy05,KLY06,GP02}. Note that by definition the critical clique graph $\C (G)$ is a block graph.

\begin{theorem}[Theorem~8~\cite{MW15}] 
\label{thm:caracBD}
Let $G = (V,E)$ be a graph. The following properties are equivalent:
    \begin{enumerate}
        \item $G$ is a \bdsc\ graph,
        \item $G$ does not contain any dart, gem or hole as an induced subgraph (see \cref{fig:obs}),
        \item $G$ is chordal and the minimal separators of $G$ are pairwise disjoint.
    \end{enumerate}
\end{theorem}

We consider the following problems:

\Pb{\textsc{Block Graph Editing}}%
{A graph $G=(V,E)$, an integer $k \in \mathbb{N}$}%
{a set of pairs $F \subseteq (V \times V)$ of size at most $k$ such that the graph 
$H = (V, E \triangle F)$ is a block graph, with $E \triangle F = (E \setminus F) \cup (F \setminus E)$}
{Question}

The \SCE{} problem is defined similarly, requiring the graph $H = (V, E \triangle F)$ to be strictly chordal. 
The \BGC{} and \SCC{} (resp. \BGD{} and \SCD{}) problems are defined similarly by requiring $F$ to be 
disjoint from (resp. included in) edge set $E$. Given a graph 
$G = (V,E)$, a set $F \subseteq (V \times V)$ such that the graph $H = (V, E \triangle F)$ belongs to the target graph class is called an \emph{edition} 
of $G$. When $F$ is disjoint from $E$ (resp. included in $E$) it is 
called a \emph{completion} (resp. a \emph{deletion}) of $G$. 
For the sake of simplicity we use $G \triangle F$, $G+F$ and $G-F$ to denote 
the resulting graphs in all versions of the problems. 
In all cases, $F$ is \emph{optimal} whenever it is minimum-sized. Given an 
instance $(G = (V,E), k)$ of any of the aforementioned problems, we 
say that $F$ is a \emph{$k$-edition} (resp. $k$-completion, $k$-deletion) whenever $F$ is an edition (resp. completion, deletion) 
of size at most $k$. A vertex is \emph{affected} by $F$ 
whenever it is contained in some pair of $F$. We say that an 
instance $(G=(V,E),k)$ is a yes-instance whenever it admits a $k$-edition (resp. $k$-completion, $k$-deletion). When applied to an instance $(I,k)$ of the problem, a reduction rule is said to be \emph{safe} if $(I, k)$ is a yes-instance if and only if the reduced instance $(I',k')$ is a yes-instance. \\

We will use the following result that guarantees that any 
clique module of a given graph $G$ will remain a clique module in any optimal edition towards some hereditary class of graphs closed under true twin addition, in particular towards \bdsc\ graphs. \\

\begin{lemma}[\cite{BPP10}]
    \label{lem:homogen-compl-here} 
    Let $\mathcal{G}$ be an hereditary class of graphs closed under true twin addition. %
    For every graph $G=(V,E)$, there exists an optimal edition (resp. completion, deletion) 
    $F$ into a graph of $\mathcal{G}$ such that for any two %
    critical cliques $K$ and $K'$ either $(K \times K') \subseteq F$ or %
    $(K \times K') \cap F = \emptyset$.
\end{lemma}
 
Notice in particular that \cref{lem:homogen-compl-here} implies that whenever 
the target graph class is hereditary and closed under true twin addition one may 
reduce the size of critical cliques to $k+1$. We shall see \cref{sec:block} that this result does not apply to block graphs. However, we will still be able to bound the size of critical cliques by $O(k)$. 

\paragraph{Join composition $\otimes$} 
We now turn our attention to the notion of \emph{join composition} (operation $\otimes$) between two graphs. This operation will play an important part in the proofs of our reduction rules for all considered problems. 
\begin{definition} 
    Let $G_1 = (V_1,E_1)$ and $ G_2 = (V_2,E_2)$ be two vertex-disjoint graphs and let $S_1 \subseteq V_1, S_2 \subseteq V_2$. The \emph{join composition} of $G_1$ and $G_2$ on $S_1$ and $S_2$, denoted $(G_1,S_1) \otimes (G_2,S_2)$, is the graph: $$(G_1,S_1) \otimes (G_2,S_2) = (V_1 \cup V_2, E_1 \cup E_2 \cup (S_1 \times S_2))$$
\end{definition} 



\begin{lemma} \label{lem:joincomp_sep}
Let $G_1 = (V_1,E_1)$ and $ G_2 = (V_2,E_2)$ be two disjoint chordal graphs, $S_1 \subseteq V_1$ and $S_2 \subseteq V_2$ be non-empty cliques of respectively $G_1$ and $G_2$ and $G= (G_1,S_1) \otimes (G_2,S_2)$. Then graph $G$ is chordal, and any minimal separator $S$ of $G$ is also a minimal separator of $G_1$ or $G_2$, except  if $S= S_i$ for $i  \in \{1,2\}$ and $S_i$ is strictly included in a maximal clique of $G_i$. 
\end{lemma}
\begin{proof}
The graph $G$ is chordal since any cycle of $G$ is either contained in one of $G_1$, $G_2$ or, if it intersects both of them, contains at least three vertices of $S_1 \cup S_2$, which is a clique in $G$. In both cases, if the cycle has four or more vertices then it contains a chord.

 Any minimal separator $S$ of chordal graph $G$ is the intersection of two maximal cliques $\Omega,\Omega'$ of $G$, and moreover $S$ is a minimal $ab$-separator for any $a \in  \Omega \setminus S$, $b \in  \Omega' \setminus S$ --- in particular, $S$ is strictly included in each of the cliques (see, e.g., Theorem~7 and Lemma~5 in~\cite{GHP95}). By construction of $G$, any of its maximal cliques is either a maximal clique of $G_1$ or of $G_2$, or it is exactly $S_1 \cup S_2$. 

If $S = \emptyset$, then graph $G$ is not connected. This only happens if one of the graphs $G_i$ is not connected, since neither $S_1$ nor $S_2$ are empty.
In this case, $S$ is a minimal separator of $G_1$ or $G_2$.

From now on, we assume that $S$ is not empty.
Consider first the case where none of $\Omega, \Omega'$ equals $S_1 \cup S_2$. Then both $\Omega, \Omega'$ are contained in the same $G_i$, for $i \in \{1,2\}$ (otherwise their intersection $S$ would be empty). Assume w.l.o.g they are both in $G_1$, we claim that $S = \Omega \cap \Omega'$ is a minimal separator of $G_1$. Take $a \in  \Omega \setminus S$ and $b \in  \Omega' \setminus S$. Since $S$ separates $a$ and $b$ in $G$, it also separates them in $G_1$, which is a subgraph of $G$. But $S$ is in the common neighborhood of $a$ and $b$ in $G_1$, thus $S$ is minimal among the $ab$-separators of $G_1$, and the conclusion follows.

We are left with the case when $\Omega = S_1 \cup S_2$ and $\Omega'$ is contained in one of the two graphs $G_i$, say in $G_1$.
Therefore $S \subseteq S_1$. If this inclusion is strict, we claim that $S$ is also a minimal separator of $G_1$. Indeed $S$ is a minimal separator in $G$ for any $a \in  \Omega \setminus S$ and $b \in  \Omega' \setminus S$, therefore we can choose $a \in S_1 \setminus S$. Again $S$ separates $a$ and $b$ in the subgraph $G_1$ of $G$, and since $S$ is in the common neighborhood of $a$ and $b$ in $G_1$, it is necessarily a minimal separator of $G$. We deduce that, if $S$ is not a minimal separator of $G_1$, we must have $S = S_1$. We conclude that $S$ is strictly contained in the maximal clique  $\Omega'$ of $G_1$, which proves our lemma.
\end{proof}

We now present sufficient conditions on the sets $S_1$ and $S_2$ so that if the graphs $G_1$ and $G_2$ are block or strictly chordal, then so is the graph $(G_1,S_1) \otimes (G_2,S_2)$.

\begin{lemma}
\label{lem:constructionBG}
Let $G_1 = (V_1,E_1)$ and $ G_2 = (V_2,E_2)$ be two disjoint block graphs and let $S_1 \subseteq V_1, S_2 \subseteq V_2$. The graph $G= (G_1,S_1) \otimes (G_2,S_2)$ is a block graph if for $i \in \{1,2\}$, $S_i$ is a single vertex or a maximal clique of $G_i$.
\end{lemma}
\begin{proof} Observe that minimal separators of block graphs are single vertices since the biconnected components of block graphs are cliques. From \cref{lem:joincomp_sep}, $G$ is chordal and its minimal separators are the ones of $G_1$ and $G_2$ plus possibly $S_1$ (resp. $S_2$) if it is a not maximal clique of $G_1$ (resp. $G_2$). Hence minimal separators of $G$ are single vertices and it follows that $G$ is a block graph.
\end{proof}

\begin{lemma}
\label{lem:constructionSC}
Let $G_1 = (V_1,E_1)$ and $ G_2 = (V_2,E_2)$ be two disjoint strictly chordal graphs and let $S_1 \subseteq V_1, S_2 \subseteq V_2$. The graph $G= (G_1,S_1) \otimes (G_2,S_2)$ is strictly chordal if for $i \in \{1,2\}$, $S_i$ is a critical clique, a maximal clique or intersects exactly one maximal clique of $G_i$. 

\end{lemma}

\begin{proof}
From \cref{lem:joincomp_sep}, $G$ is chordal and its minimal separators are the ones of $G_1$ and $G_2$ plus possibly $S_1$ or $S_2$ if they are not maximal cliques. If $S_i$ intersects exactly one maximal clique, $i\in\{1,2\}$, it is clear that it does not intersect any minimal separator of $G_i$. If $S_i$ is a critical clique, we claim that either $S_i$ is a minimal separator of $G_i$ or does not intersect any minimal separator of $G_i$. Indeed, if there is a minimal (clique) separator $S$ of $G_i$ such that $S_i \subset S$, then since $S_i$ is a critical clique, there exist $u\in S_i, v\in S\setminus S_i$ and $w\in V(G_i)\setminus S$ such that $w$ is adjacent to exactly one of $u$ and $v$. Suppose w.l.o.g. that $w$ is adjacent to $u$ and not to $v$, then $u$ is in a minimal $vw$-separator that intersects $S$ and is not equal to it. This is a contradiction since $G_i$ is strictly chordal and by \cref{thm:caracBD} its minimal separators are pairwise disjoint. It follows that the minimal separators of $G$ are pairwise disjoint and by \cref{thm:caracBD} $G$ is strictly chordal.  
\end{proof}

In the proofs of some rules for the kernel for \SCE{} and its variants, we will have to do a join composition of some graphs $G_2, \ldots, G_r$ on a graph $G_1$ with the same set of vertices $S_1\subseteq V(G_1)$. The following observation guarantees us that this operation is possible.

\begin{observation} \label{obs:SCmulticomp}
Let $G_1 = (V_1,E_1), \dots, G_r = (V_r,E_r)$ be disjoint strictly chordal graphs and for each $i\in\{1,\dots , r\}$ let $S_i \subseteq V_i$ be a critical clique, a maximal clique or a set that intersects exactly one maximal clique of $G_i$. If $S_1$ is a critical clique or intersects exactly one maximal clique of $G_1$, then $S_1$ is a critical clique or intersects exactly one maximal clique of $(G_1,S_1) \otimes (G_2,S_2)$. This implies by \cref{lem:constructionSC} that $(G_1, S_1) \otimes (\bigcup_{2 \leq i \leq r}G_i, \bigcup_{2 \leq i \leq r}S_i)$ is strictly chordal (see \cref{fig:BG_join_multicomp}). 
\end{observation}

\begin{figure}[ht!]
    \centering
    \includegraphics[scale=1.3]{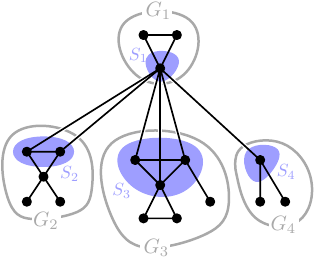}
    \caption{Example of a graph $G= (G_1, S_1) \otimes (\bigcup_{2 \leq i \leq 4}G_i, \bigcup_{2 \leq i \leq 4}S_i)$ as constructed in \cref{obs:SCmulticomp}. There, $S_1$ intersects exactly one maximal clique of $G_1$ and is a critical clique of $G$; $S_2$ intersects exactly one maximal clique of $G_2$; $S_3$ is a maximal clique of $G_3$; $S_4$ is a critical clique of $G_4$.} 
    \label{fig:BG_join_multicomp}
\end{figure}

\newcommand{\sd}[1]{d(#1)}
\paragraph{Spanning block subgraph} The following result is crucial to bound the size of the kernels for all variants of the studied modification problems. For a graph $G$, let $\sd{G}$ be the set of vertices of degree at least $3$ in $G$.

\begin{lemma} \label{lem:bloc-span}
Let $G= (V,E)$ be a connected block graph, $A$ be a subset of $V(G)$ and $T_A$ be a minimal connected induced subgraph of $G$ that spans all vertices of $A$. Denote by $\sd{T_A}$ the set of vertices of degree at least $3$ in $T_A$. The following holds:
\begin{enumerate}[(i)]
    \item\label{lem:bloc-spani} The subgraph $T_A$ is unique,
    \item\label{lem:bloc-spanii} $|\sd{T_A}| \leq 3 \cdot |A|$, 
    \item\label{lem:bloc-spaniii} The graph $T_A \setminus (A\cup \sd{T_A})$ contains at most $2\cdot |A|$ connected components.
\end{enumerate}
\end{lemma}

\begin{proof} 
We first prove the uniqueness of $T_A$. Observe that the internal vertices of a shortest path $P$ in $G$ between two vertices $u$ and $v$ are cut-vertices of $G$, since it is a block graph. Moreover, for any such internal vertex $w$ of $P$, vertices $u$ and $v$ are in different components of $G-w$. 
Therefore, $P$ is unique, and any path between $u$ and $v$ must contain every vertex of $P$. Hence, the vertices of $T_A$ are either in $A$ or separating vertices of $A$, and the uniqueness of $T_A$ follows, proving \cref{lem:bloc-spani}.

To prove \cref{lem:bloc-spanii} we will proceed by induction on the size of $A$. 
For $|A| \leq 2$ we have that $|\sd{T_A}|=0$. Suppose now that $|A| \geq 3$.
Let  $a\in A$ be a non-cut-vertex of $T_A$. Such a vertex exists, e.g., it is sufficient to choose a longest path of $T_A$, and let $a$ be one of the endpoints. Then $a$ is not a cut vertex of $T_A$ (as the end of a longest path), and $a \in A$ by minimality of $T_A$. Notice moreover that $a$ is in only one maximal clique of $T_A$. 
Let $A' =  A\setminus\{a\}$ and let $T_{A'}$ be a minimal connected induced subgraph of $G$ that spans all vertices of $A'$. 
Observe that $T_{A'}$ is an induced subgraph of $T_A$. 
We show that $|\sd{T_A}| \leq |\sd{T_{A'}}|+3$. 
Since $T_A$ is connected, there exists a path in $T_A$ from $a$ to a vertex $u \in V(T_A)\setminus V(T_{A'})$ adjacent to some vertices of $T_{A'}$ ($a=u$ is possible), moreover $u$ is unique. 
Any vertex on the path except for $u$ is of degree at most $2$ in $T_A$. 
Let $N= N_G(u) \cap V(T_{A'})$ and note that since $T_{A'}$ is a connected block graph, $N$ is a clique. 
By construction, the vertices of $T_{A'}$ have the same degree in $T_A$ and $T_{A'}$, except for the ones in $N$. 
Since $u$ might be of degree $3$ or more, if $|N| \leq 2$ then $|\sd{T_A}| \leq |\sd{T_{A'}}|+3$. 
If $|N| \geq 4$, then all vertices in $N$ already have degree $3$ in $T_{A'}$, hence $|\sd{T_A}| = |\sd{T_{A'}}|+1$. 
If $|N|=3$ and $N$ contains a vertex of degree $3$ in $T_{A'}$ then $|\sd{T_A}| \leq |\sd{T_{A'}}|+3$. 
Otherwise, we claim that $V(T_{A'}) = N$. Indeed, if all vertices in $N$ have degree $2$ in $T_{A'}$, 
the only possibility is that $N$ is exactly the vertex set of $T_{A'}$, because $T_{A'}$ is connected, $N$ induces a clique of size $3$ in this graph and no vertex of $N$ has other neighbors than the ones in $N$.
In this case we have $|A| = 4$ and $|\sd{T_A}| = 4$, which satisfies \cref{lem:bloc-spanii}.
Therefore, except for the latter case, we have $|\sd{T_A}| \leq |\sd{T_{A'}}|+3$, and using the induction hypothesis we have $|\sd{T_A}| \leq |\sd{T_{A'}}|+3 \leq 3 \cdot |A'| + 3 = 3 \cdot |A|$.

We are left with \cref{lem:bloc-spaniii}, which can be proven in a similar way as the previous one. Using the same notation, 
 the only vertices of $T_{A'}$ whose degree may increase in $T_A$ are in $N$. Therefore, every connected component of $T_{A'} \setminus (A'\cup \sd{T_{A'}})$ that does not contain  vertices of $N$ is also a connected component of $T_A \setminus (A\cup \sd{T_A})$. Moreover, a connected component of $T_{A'} \setminus (A'\cup \sd{T_{A'}})$ may contain vertices of $N$ only if these vertices have degree at most $2$ in $T_{A'}$ and are not in $A$. In particular $|N| \leq 2$
 and both vertices lie in the same connected component of $T_{A'} \setminus (A'\cup \sd{T_{A'}})$ (if any).
Therefore, this connected component may be split in two disjoint ones in $T_A \setminus (A\cup \sd{T_A})$ if the vertices of $N$ are removed. Moreover, an additional connected component corresponding to the path between $a$ and $u$ may exist in $T_A$. Hence, we can only increase the number of connected components of $T_A \setminus (A\cup \sd{T_A})$ by  compared to $T_{A'} \setminus (A'\cup \sd{T_{A'}})$. Since the property is verified for $|A|=1$, it follows by induction that there are $2\cdot|A|$ connected components in $T_A \setminus (A\cup \sd{T_A})$. 
\end{proof}

\subsection{Hardness results}

We first show the NP-completeness of \BGD{} and \BGE{} by giving a reduction from \CD{} and \CE{}, known to be NP-complete~\cite{KM86,SST04,Natanzon99}. Notice that \BGC{} admits 
a polynomial-time algorithm due to the very nature of block graphs: 
any biconnected component  of the input graph (computable in linear time~\cite{HT73}) must be turned into a clique. 

\begin{lemma}
\BGD{} and \BGE{} are NP-Complete.
\end{lemma}

\begin{proof}
We first prove that \BGD{} is NP-Complete by providing a reduction from \CD. A graph is a cluster graph if it does not contain any induced path on three vertices (so-called $P_3$). Given an instance $(G=(V,E),k)$ of \CD{}, we construct an instance of \BGD{} by adding a universal vertex $u$ adjacent to all vertices of $V$. Let $(G'=(V',E'),k)$ be the produced instance. 

We show that the graph $G$ admits a $k$-deletion into a cluster graph if and only if $G'$ admits a $k$-deletion into a block graph. Suppose first that there is a $k$-deletion $F$ of $G$ into a cluster graph. The graph $G - F$ is a graph without any $P_3$ as induced subgraph. Now consider the graph $H' = G' - F$. By construction $H'[V]$ contains no $P_3$, so $H'$ is chordal and contains no diamond, thus $H'$ is a block graph. 

Now suppose that there exists a $k$-deletion $F'$ of $G'$ into a block graph. 
We claim that either $F'$ does not affect the universal vertex $u$ or we can construct a $k$-deletion $F^*$ from $F'$ that does not affect $u$. 
To support this claim, assume that $F'$ contains at least one edge incident to $u$.

First, suppose that $H' = G' - F'$ is not connected. Let $C$ be the connected component that contains $u$. For any other connected component $C'$ of $H'$, take a vertex $v\in V(C')$ and construct the deletion set $F'' = F'\setminus \{(u,v)\}$. By \cref{lem:constructionBG}, $(C, \{u\}) \otimes (C',\{v\})$ is a block graph, and the other connected components of $H'$ are unchanged. Hence $G' - F''$ is a block graph and since $|F''| < |F'|$, $F''$ is a $k$-deletion. We iterate this construction until we get an intermediate solution $F^*_1$ such that $H^*_1 = G' - F^*_1$ is connected. 

If $F^*_1$ does not affect $u$ we set $F^*=F^*_1$ and we are done. Otherwise there is at least one vertex $v$ not adjacent to $u$ in $H^*_1$. Since $H^*_1$ is connected we can choose $v$ at distance exactly $2$ from $u$. Since $H^*_1$ is a block graph we have $N(u) \cap N(v)=\{w\}$. 
Consider the biconnected component $C$ containing $v$ and $w$, 
and let $C'=C\setminus \{w\}$. Since $H^*_1$ is a block graph, $C$ (and thus $C'$) is a clique. Let $K_u$ and $K_v$ be the connected components of $H^*_1 - (\{w\} \times C')$ that contain the vertices $u$ and $v$ respectively, they are block graphs by heredity. By \cref{lem:constructionBG} the graph $H^*_2 = (K_u, \{u\}) \otimes (K_v,\{C'\})$ is a block graph. Observe that $(\{u\} \times C')\subseteq F^*_1$ since $u$ is an universal vertex and let $F^*_2$ be the deletion set such that $H^*_2 = G' - F^*_2$. By construction we can observe that $|F^*_2| = |F^*_1|$, hence $F^*_2$ is a $k$-deletion of $G'$ into a block graph. We iterate this construction until we get a $k$-deletion $F^*$ of $G$ that does not affect $u$. 

Finally, $H = G - F^*$ does not contains any $P_ 3$ or else there would be a diamond in $G' - F^*$ with the universal vertex $u$, thus $H$ is a cluster graph.

The same reduction can be done for \BGE{}, one needs only to observe that there are no added edges incident to $u$, and we can use the same argument as above on the deleted edges: if there is no deleted edge incident to $u$ we are done, otherwise we construct another set of deleted edges of equal size, but containing fewer edges incident to $u$. 
If in the construction we have to delete an edge that was added by the edition, it is clear that removing this edge from the edition set result in a strictly smaller edition set.
\end{proof}

The NP-completeness of \SCC{} follows directly from the proof of NP-completeness of \LPC{} from~\cite[Theorem 3]{DGH+06}.

\begin{lemma}
\SCC{} is NP-complete.
\end{lemma}

We show the NP-completeness of \SCE{} and \SCD{} by giving a reduction from \CE{} and \CD{}, respectively. 

\begin{lemma}
\SCE{} and \SCD{} are NP-complete.
\end{lemma}

\begin{proof}
Given an instance $(G=(V,E),k)$ of \CE{}, we construct an instance of \SCE{} by adding a clique $U = \{u_1, \dots , u_{k+1}\}$ of size $k+1$ adjacent to all vertices of $V$, and for each vertex $x$ in $V$, $k+1$ vertices $\{v^x_1,\dots ,v^x_{k+1}\}$ adjacent only to $x$. Let $(G'=(V',E'),k)$ be the produced instance. 
We show that the graph $G$ admits a $k$-edition into a cluster graph if and only if $G'$ admits a $k$-edition into a strictly chordal graph. Suppose first that there is a $k$-edition $F$ of $G$ into a cluster graph. The graph $G \triangle F$ is a graph without any $P_3$ as induced subgraph. Consider the graph $H' = G' \triangle F$. By construction $H'[V]$ contains no $P_3$, so $H'$ is chordal and contains neither gems nor darts since these obstructions contain an induced $P_3$ and a vertex adjacent to every vertex of this $P_3$. By \cref{thm:caracBD}, it follows that $H'$ is strictly chordal. Now suppose that there exists a $k$-edition $F'$ of $G'$ into a strictly chordal graph $H' = G' \triangle F'$. 
We claim that $H = G \triangle F'$ is a cluster graph. By contradiction, suppose that $H = H'[V]$ contains a $P_3$ $\{x,y,z\}$ where $x,z$ are the ends of the path. Then, there exist $i,j \in \{1, \dots , k+1\}$ such that $\{x,y,z,u_i,v^y_j\}$ forms a dart in $H'$, contradicting that $H'$ is strictly chordal, and thus $H$ is a cluster graph.

The same reduction can be done from \CD{} to \SCD{}. This concludes the proof. 
\end{proof}

\section{Kernelization algorithm for {\sc \BGE{}}}
\label{sec:block}

We begin this section by providing a high-level description of our kernelization algorithm. A similar technique will be applied for modification problems towards 
strictly chordal graphs but on the critical clique graph rather than the original one. Let us consider a positive instance $(G = (V,E),k)$ of \BGE{}, $F$ a suitable solution and $H = G \triangle F$. Since $|F| \leq k$, we know that at most $2k$ vertices of $H$ may be affected vertices. 
Let $A$ be the set of such vertices, $T$ the minimum induced subgraph of $H$ that spans all vertices of $A$ and  $A'$ the set of vertices of degree at least $3$ in $T$. 
From \cref{lem:bloc-span} we have $|A'\setminus A| \leq 3\cdot |A|$. We will define the notion of BG-branch, corresponding to subgraphs of $G$ that induce block graphs. We will focus our interest on two types of BG-branches: the ones that are connected to the rest of the graph by only one (cut) vertex, called \emph{$1$-BG-branches}, and the ones that are connected to the rest of the graph by exactly two non-adjacent vertices, called \emph{$2$-BG-branches}. We will prove that we can keep only the cut-vertices of $1$-BG-branches and reduce the $2$-BG-branches containing more than $k+3$ vertices. 
We will see that there are two kinds of connected components remaining in the graph $H \setminus (A\cup A')$, the ones adjacent to maximal cliques of vertices of $A\cup A'$, corresponding to critical cliques (which can be bounded to a linear number of vertices) and the ones adjacent to two non-adjacent vertices of $A\cup A'$ (which correspond to $2$-BG-branches in $H$). Since $H$ is a block graph, we know that there is at most $|A\cup A'| = O(k)$ maximal cliques of vertices of $A\cup A'$. From \cref{lem:bloc-span} we have that there are at most $4k$ connected components in $T\setminus (A \cup A')$, and each one corresponds to a $2$-BG-branch in the graph $H$. It remains that $H \setminus (A\cup A')$ contains a linear number of connected components, each one containing $O(k)$ vertices. Altogether, the graph $H$ contains $O(k^2)$ vertices.

\subsection{Classical reduction rules}
We first give classical reduction rules when dealing with modification problems. Notice that \cref{lem:homogen-compl-here} cannot be 
directly applied to \BGE{} and \BGD{} since block graphs are not closed under true twin addition. Hence one cannot directly reduce the size of critical cliques to $k+1$ as mentioned \cref{sec:prelim}. To circumvent this issue, we provide a slightly weaker reduction rule. 

\begin{polyrule} \label{rule:compBG}
    Let $C$ be a \bg{} connected component of $G$. 
    Remove $V(C)$ from $G$.
\end{polyrule}

\begin{polyrule}\label{rule:borneCC-BG}
Let $K \subseteq V$ be a set of true twins of $G$ such that $|K| > 2k+2$. Remove $|K|-(2k+2)$ arbitrary vertices from $K$ in $G$. 
\end{polyrule}

\begin{lemma}\label{lem:safecc-BG}
\cref{rule:compBG,rule:borneCC-BG} are safe and can be computed in linear time. 
\end{lemma}

\begin{proof}
\cref{rule:compBG} is safe since block graphs are hereditary and closed under disjoint union.

We now turn our attention to \cref{rule:borneCC-BG}. Let $G'$ denote the 
graph obtained from the removal of $|K|-(2k+2)$ arbitrary vertices from $K$ 
in $G$ and $K' = V(G) \setminus V(G')$. 
Assume first that $G$ admits a $k$-edition $F$. Observe that 
$F$ is a $k$-edition of $G'$ by heredity. Conversely, assume that $G'$ admits 
a $k$-edition $F$ and let $H' = G' \triangle F$. 
Notice that since $F$ affects at most $2k$ vertices and since 
$|K \setminus K'| = 2k+2$, there 
exist two unaffected vertices $u$ and $v$ in $K \setminus K'$. 
We claim that $H = G \triangle F'$ is a block graph. Assume for a contradiction that $H$ contains an obstruction $W$. 
By construction, $W$ must intersect $K'$ since otherwise $W$ would exist in $H'$. Moreover, notice that 
$|W \cap K'| \leq 2$, otherwise $W$ would contain a clique module 
with $3$ or $4$ vertices, which is impossible in a cycle or 
a diamond. Hence assume first that $|W \cap K'| = \{a\}$: in this case, since $u$ is unaffected we have $N_{H}[a]\setminus K' = N_{H'}[u]$ and thus the 
set $(W \setminus \{a\}) \cup \{u\}$ induces an obstruction in $H'$, a contradiction. Similarly, if $|W \cap K'| = \{a,b\}$ the set 
$(W \setminus \{a,b\}) \cup \{u,v\}$ induces an obstruction in $H'$, leading 
once again to a contradiction. 
Finally, since the biconnected components can be computed in $O(n+m)$ time~\cite{HT73}, since block graphs can be recognized in $O(n+m)$ time and since true twins can be detected in $O(n+m)$ time (using modular decomposition~\cite{PDS09}), both rules can be computed in linear time. 
\end{proof}

\subsection{Reducing Block Graph branches}

Let $G=(V,E)$ be a graph, a \emph{BG-branch} is a connected induced subgraph $B$ of $G$ such that $B$ is a block graph. A vertex $v\in V(B)$ is an \emph{attachment point} of $B$ if $N_G(v)\setminus V(B) \neq \emptyset$. A BG-branch is a $p$-BG-branch if it has exactly $p$ attachment points. 
We denote $B^R$ the induced subgraph of $B$ in which all attachment points have been removed.
We first define a simple rule that allows to remove all the vertices of a $1$-BG-branch except for its attachment point.

\begin{polyrule} \label{rule:1BGbranch} 
    Let $B$ be a $1$-BG-branch of $G$ such that $B^R$ is connected and not empty. 
    Remove $V(B^R)$ from $G$.
\end{polyrule}

\begin{lemma}
\cref{rule:1BGbranch} is safe. 
\end{lemma}

\begin{proof}
Let $B$ be such a $1$-BG-branch and $G'=G\setminus V(B^R)$.
Let $F$ be a $k$-edition of $G$ into a block graph.
Since the class of block graphs is hereditary, $(G\triangle F) \setminus V(B^R)$ is a block graph. It follows that $G'\triangle F$ is a block graph and $F$ is a $k$-edition of $G'$ into a block graph.

Conversely, let $F'$ be a $k$-edition of $G'$ into a block graph. Let $u$ be the attachment point of $B$ in $G$. Observe that since $B^R$ is connected and not empty, $u$ is not a cut-vertex of the block graph $B$ and its neighborhood in $B$ is either a maximal clique or a single vertex. 
Since $H' = G' \triangle F'$ and $B$ are block graphs and $N_B(u)$ is a single vertex or a maximal clique of $B$, Lemma \ref{lem:constructionBG} ensures that $H = (H', \{u\})\otimes(B^R, N_B(u))$ is a block graph. We can observe that $H = G \triangle F'$, hence $F'$ is a $k$-edition of $G$ into a block graph. This concludes the proof. 
\end{proof}

We now deal with the $2$-BG-branches, which require more technical skills. 
A $2$-BG-branch $B$ is \emph{clean} if $B^R$ is connected and if the two attachment points are not adjacent. For a clean $2$-BG-branch $B$, a \emph{min-cut} of $B$ is a set of edges $M \subseteq E(B)$ of minimum size such that the attachment points $p_1$ and $p_2$ are not in the same connected component of $B-M$ and $mc(B)$ is the size of a min-cut of $B$. 

\begin{lemma} \label{lem:2BGbranch}
Let $(G,k)$ be a yes-instance of \BGE{} reduced by Rule \ref{rule:1BGbranch}, and $B$ a clean $2$-BG-branch of $G$. If $|V(B)|\geq k+3$ then there exists a $k$-edition $F$ of $G$ into a block graph such that the set of pairs of $F$ containing a vertex of $V(B^R)$ is either empty or exactly a min-cut of $B$. 
\end{lemma}

\begin{proof}
Let $B$ be a clean $2$-BG-branch in $G$ with attachment points $p_1,p_2$ such that $|V(B)|\geq k+3$. We can observe that $N_B(p_1)$ and $N_B(p_2)$ are maximal cliques or single vertices since $B^R$ is connected. Since $B$ is reduced by \cref{rule:1BGbranch}, we can observe that the cut-vertices of $B$ are contained in any path $\pi$ from $p_1$ to $p_2$. 
Indeed, if there is a cut-vertex $x$ that is not in $\pi$, then $p_1$ and $p_2$ are in the same connected component of $B\setminus \{x\}$, implying that there is a connected component $K$ of $B\setminus \{x\}$ containing neither $p_1$ nor $p_2$. Since $B$ is a $2$-BG-branch, $N_G(K) = \{x\}$, hence $G[V(K)\cup \{x\}]$ induces a $1$-BG-branch of $G$ with attachment point $x$, a contradiction. Notice in particular that there is a unique shortest path $\pi_s$ between $p_1$ and $p_2$ that contain all cut-vertices of $B$. 
Moreover, since the only cut-vertices of $B$ are the ones in $\pi_s$, a vertex of $B$ that is not in $\pi_s$ lies in exactly one maximal clique of $B$ containing two consecutive vertices of $\pi_S$. 

\begin{figure}[h!]
    \centering
     \includegraphics[scale=2.1]{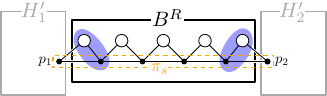}
     \caption{Illustration of the case where $p_1$ and $p_2$ are not in the same connected component of $H'$ in the \cref{lem:2BGbranch}. The white circles correspond to critical cliques of $B$. The blue sets correspond to $N_B(P_1)$ and $N_B(P_2)$.}
    \label{fig:2_BG_connected}
\end{figure}

Let $F$ be a $k$-edition of $G$ into a block graph, $H=G \triangle F$ and $H'=H\setminus V(B^R)$. First consider the case in which $p_1$ and $p_2$ are in two different connected components in $H'$. Let $H'_1, H'_2$ be the two connected components of $H'$ containing $p_1$ and $p_2$, respectively, and $H'_3$ the remaining of $H'$.  It follows from Lemma \ref{lem:constructionBG} that the graph $H^*$ resulting from the disjoint union of $$(((B^R, N_B(p_1))\otimes (H'_1, \{p_1\})), N_B(p_2)) \otimes (H'_2, \{p_2\})$$ and $H'_3$ is a block graph (see \cref{fig:2_BG_connected}). 
Let $F^*$ be the edition such that $H^*=G \triangle F^*$. By construction, $F^*$ is a subset of $F$ with no pair containing a vertex of $V(B^R)$.

Assume now that $p_1$ and $p_2$ are in the same connected component in $H'$ and let $\pi_{H'}$ denote a shortest path between them. Note that $\pi_{H'}$ may be of length $1$ if $p_1$ and $p_2$ are adjacent in $H'$. We first consider the case where there is a path $\pi_B=\{p_1=u_1,u_2,\dots,u_r=p_2\}$ (not necessarily induced) from $p_1$ to $p_2$ in $B$ that still exists in $B \triangle F$. Since $B$ is a clean $2$-BG-branch of $G$ reduced by \cref{rule:1BGbranch}, we know by the previous arguments that every cut-vertex of $B$ is in $\pi_B$. Note that all vertices in $\pi_B$ belong to the same biconnected component of $H$, which is a clique. Hence we can consider w.l.o.g. that $\pi_B$ is the path containing exactly $p_1, p_2$ and the cut-vertices of $B$, which is the shortest path in $B$ between $p_1$ and $p_2$. Since the vertices of $\pi_B$ induce a path in $B$ (on at least $3$ vertices since $p_1$ and $p_2$ are not adjacent) and a clique in $B \triangle F$, then for every $i\in \{1,\dots,r-2\}$, we have $(u_i, u_r)\in F$. 
Any vertex $v \in B\setminus V(\pi_B)$ is adjacent in $B$ to two vertices $w_1, w_2\in V(\pi_B)$ and if $v$ remains in the same biconnected component as $w_1$ and $w_2$ in $H$, then $(v,p_1)\in F$ or $(v,p_2) \in F$. Otherwise $F$ must contain $(v,w_1)$ or $(v,w_2)$. In any case, $F$ contains at least one pair for each vertex $v \in B\setminus \{u_{r-1}, u_r\}$. 
It follows from the previous arguments that $|F|\geq|B|-2$. However, since $|B|\geq k+3$ we have $|F|\geq k+1$, contradicting that $F$ is a solution, therefore there is no such path $\pi_B$.


We can assume that if $p_1$ and $p_2$ are in the same component of $H'$,  $F$ contains an edge-cut of $B$ or there would be a path $\pi_B$ from $p_1$ to $p_2$ in $B$ that still exists in $B\triangle F$. Let $M$ be a min-cut of $B$ and consider $B_1$ (resp. $B_2$) the connected component containing $p_1$ (resp. $p_2$) in $B \triangle M$.  Both $B^R_1=B_1\setminus\{p_1\}$ and $B^R_2=B_2\setminus\{p_2\}$ are induced subgraphs (one of them possibly empty) of $B$ and block graphs by heredity. By Lemma \ref{lem:constructionBG}, the following graph is a block graph: $$H^*=((H',\{p_1\})\otimes(B^R_1, N_{B_1}(p_1)),\{p_2\})\otimes (B^R_2, N_{B_2}(p_2))$$ 

\begin{figure}[h]
    \centering
    \includegraphics[width=6.5cm]{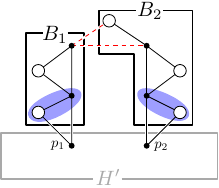}
    \caption{Illustration of the case where $p_1$ and $p_2$ are in the same connected component of $H'$ and $F$ contains a min-cut of $B$ in \cref{lem:2-SC-branchE}. 
    The white circles correspond to critical cliques of $B$. The blue sets correspond to $N_B(P_1)$ and $N_B(P_2)$. The red dotted edges between $B_1$ and $B_2$ are edges of an edge-cut of $B$ deleted by the edition.
    }
    \label{fig:2BG-branch}
\end{figure} 

Let $F^*$ the edition such that $H^* = G\triangle F^*$, we have $F^* = (F \setminus (V(B^R) \times V(G))) \cup M$, since $M$ is a min-cut of $B$, $|F^*| \leq |F|$ and $F^*$ is a $k$-edition of $G$.

In all cases, there exists a $k$-edition $F^*$ of $G$ into a block graph such that the set pairs of $F^*$ containing a vertex of $V(B^R)$ is either empty or exactly a min-cut of $B$. This concludes the proof. 
\end{proof}

\begin{polyrule} \label{rule:2BGbranch}  
Let $(G=(V,E),k)$ be an instance of \BGE{} and $B$ be a clean $2$-BG-branch of $G$ containing at least $k+3$ vertices with attachment points $p_1,p_2$. Remove $V(B^R)$ from $G$ and add a vertex $x$ adjacent to $p_1$ and $p_2$, a clique of size $min\{mc(B) -1, k\}$ adjacent to $p_1$ and $x$ and a clique of size $k$ 
adjacent to $p_2$ and $x$.
\end{polyrule}

\begin{lemma}
\cref{rule:2BGbranch} is safe.
\end{lemma}

\begin{proof}
Let $G'$ be the graph obtained from an application of \cref{rule:2BGbranch} on the $2$-BG-branch $B$ of graph $G$ with attachment points $p_1,p_2$ and $S$ the vertices introduced by this rule. Observe that $B' = G'[S \cup \{p_1\} \cup \{p_2\}]$ is a clean $2$-BG-branch of $G'$,  $|V(B')| \geq k+3$ and $mc(B') = min\{mc(B) , k+1\}$.

Let $F$ be a $k$-edition of $G$ that satisfies \cref{lem:2BGbranch} and $H = G \triangle F$. If $F$ contains a min-cut of $B$, observe that $mc(B') = mc(B) \leq k$ and consider $B_i', i\in \{1,2\}$ the connected component of $B' \triangle F_m$ containing $p_i$, where $F_m$ is a min-cut of $B'$. By \cref{lem:constructionBG} the graph $$H' = ((H\setminus V(B^R), p_1) \otimes (B_1',N_{B'}(p_1)), p_2) \otimes (B_2',N_{B'}(p_2))$$ is a block graph. Let $F'$ be the edition such that $H'  = G' \triangle F'$. Since $mc(B') = mc(B) \leq k$, by construction $|F'| = |F|$ , therefore $F'$ is a $k$-edition of $G'$. If $F$ does not contain a min-cut of $B$, then no vertex of $B^R$ is affected and $p_1, p_2$ are not in the same connected component of $H\setminus V(B^R)$. Let $H_1$ and $H_2$ be the connected components of $H\setminus V(B^R)$ that contain respectively $p_1$ and $p_2$ and $H_3$ the remaining components of $H\setminus V(B^R)$. By \cref{lem:constructionBG} the graph $H'$ corresponding to the disjoint union of $$(((B'^R, N_{B'}(p_1))\otimes (H_1, \{p_1\})), N_{B'}(p_2)) \otimes (H_2, \{p_2\})$$ 
and $H_3$ is a block graph.  We can observe that $H' = G' \triangle F$, hence $F$ is a $k$-edition of $G'$ into a block graph. 
The other way is symmetrical. 
\end{proof}

\subsection{Bounding the size of the kernel}

To conclude our kernelization algorithm we first state that reduction rules involving BG-branches can be applied in polynomial time.

\begin{lemma}\label{lem:branchBG}
\cref{rule:1BGbranch,rule:2BGbranch} can be 
applied in polynomial time. 
\end{lemma}

\begin{proof}
We show that we can enumerate all $1$-BG-branches and $2$-BG-branches in polynomial time. Recall that recognition of block graphs can be computed in linear time. All $1$-BG-branches can be enumerated by removing a vertex $v\in V(G)$ and looking among the remaining connected components those that induce a connected block graph together with $v$. 
We proceed similarly to detect clean $2$-BG-branches by removing a pair of non-adjacent vertices $u,v\in V(G)$ and looking among the remaining connected components those that induce a connected block graph together with $u$ and $v$.
\end{proof}

\THMBGE*

\begin{proof}
Let $(G=(V,E),k)$ be a reduced yes-instance of \BGE, $F$ a $k$-edition of $G$ and $H=G \triangle F$. We assume that $G$ is connected, the following arguments can be easily adapted if this is not the case by summing over all connected components of $G$. Let $A$ be the set of affected vertices of $H$. Since $|F| \leq k$, we have $|A| \leq 2k$. Let $T_A$ be the connected minimal induced subgraph of $H$ that spans $A$, and $\sd{T_A}$ the set of vertices of degree at least $3$ in $T_A$. By \cref{lem:bloc-span}, $|\sd{T_A} \setminus A| \leq 3\cdot |A| \leq 6k$.

First, we can observe that if there is a cut-vertex $x\in V(H)$ that is not in $V(T_A)$, since $T_A$ is the connected minimal induced subgraph of $H$ that spans all affected vertices, there is a connected component $K$ of $H\setminus \{x\}$ that does not contain any affected vertex. Since $B = H[V(K)\cup \{x\}]$ is a block graph by heredity, it follows that $B$ induces a $1$-BG-branch with attachment point $x$ and can be reduced by \cref{rule:1BGbranch}, a contradiction. It remains that every cut-vertex of $H$ lies in $V(T_A)$ and that every vertex of $V(H) \setminus V(T_A)$ is adjacent to a maximal clique of vertices or a single vertex of $V(T_A)$. Since the graph is reduced under \cref{rule:1BGbranch}, the second case can not happen or else there would a $1$-BG-branch.
Hence, there are two kinds of connected components in $H\setminus (A \cup \sd{T_A})$, namely the ones that are adjacent to a maximal clique of vertices of $A \cup \sd{T_A}$ and the ones that are adjacent to two non-adjacent vertices of $A \cup \sd{T_A}$ (i.e., the ones that contain connected components of $T_A \setminus (A\cup \sd{T_A})$). We say that these connected components are of type 1 or type 2 respectively (see \cref{fig:ex_graphe_reduit_BG}).

The type 1 connected components are cliques since $H$ is a block graph and have the same neighborhood in $H$, thus they are critical cliques reduced by \cref{rule:borneCC-BG} and contains at most $2k+2$ vertices. Since $T_A$ is a block graph, there is  at most $|A \cup \sd{T_A}| \leq 8k$ maximal cliques of vertices of $A\cup \sd{T_A}$, hence there are at most  $8k \cdot (2k+2)$ 
vertices in connected components of type 1. 

Now consider the connected components of type 2. We can observe that such a component $K$ together with its adjacent vertices $x,y \in A\cup \sd{T_A}$ form a clean 2-$BG$-branch, and thus contains at most $2k+3$ vertices. By \cref{lem:bloc-span}, $T_A \setminus (A\cup \sd{T_A})$ contains at most $2\cdot |A| = 4k$ connected components which is exactly the number of type 2 components, hence there are at most $4k \cdot (2k+1)$ 
vertices in connected components of type 2. 

 We conclude that if $(G,k)$ is a reduced yes-instance, then $|V(G)| \leq 24k^2 + 28k$. Therefore, given a reduced instance $(G,k)$ of \BGE{}, if $|V(G)| > 24k^2 + 28k$ we can safely conclude that $(G,k)$ is a no-instance and return a trivial no-instance of constant size, thus proving the size of our kernel. Finally we claim that the reduced instance can be computed in polynomial time. Indeed, all reduction rules can be 
exhaustively applied in polynomial time (\cref{lem:safecc-BG,lem:branchBG}). 
\begin{figure}
    \centering
    \includegraphics[width=10cm]{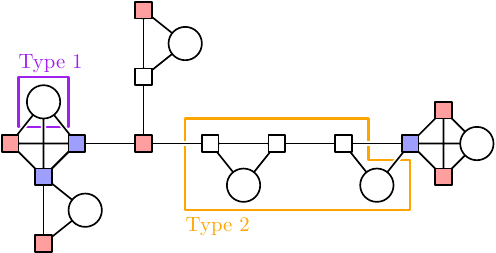}
    \caption{Illustration of a reduced instance in the proof of \cref{thm:taille_noyau_BGE}. Square nodes correspond to vertices of $T_A$, the ones filled in red are vertices of $A$, the ones filled in blue are vertices of $\sd{T_A}\setminus A$. Circle nodes correspond to critical cliques of $G$.}
    \label{fig:ex_graphe_reduit_BG}
\end{figure}
For the case of \BGD{}, all the lemmata in this section are still true. In particular \cref{lem:2BGbranch} holds by simply replacing "$k$-edition" with "$k$-deletion", and the proof for the deletion variant is contained in the edition one. Using the same set of rules, we have a kernel with $O(k^2)$ vertices for \BGD{}.
\end{proof}




\section{Kernelization algorithm for {\sc \SCE{}} }
\label{sec:sc}
As in the previous section, we begin this section by providing a high-level description of our kernelization algorithm. We use the critical clique graph of strictly chordal graphs, which is a block graph by definition, to bound the number of vertices of a reduced instance. While the technique used to bound the size of the kernels for block graphs and strictly chordal graphs is similar, the structure and rules for the latter are more involved. We thus provide a stand-alone proof of this result for the sake of completeness. 

Let us consider a positive instance $(G = (V,E),k)$ of \SCE{}, $F$ a suitable solution and $H = G \triangle F$. Denote by $\C(H)$ the critical clique graph of $H$ as described in \cref{def:ccg}. Since $|F| \leq k$, we know that at most $2k$ critical cliques of $\C(H)$ may contain affected vertices. Let $A$ be the set of such critical cliques, $T$ the minimum subgraph of $\C(H)$ that spans all critical cliques of $A$ and  $A'$ the set of critical cliques of degree at least $3$ in $T$. Recall that $|A'\setminus A| \leq 3\cdot |A|$ (\cref{lem:bloc-span}). We will define the notion of SC-branch, corresponding to well-structured strictly chordal subgraphs of $G$. We will focus our interest on two types of SC-branches: the ones that are connected to the rest of the graph by only one critical clique, called \emph{$1$-SC-branches}, and the ones that are connected to the rest of the graph by exactly two critical cliques, called \emph{$2$-SC-branches}. This time the $1$-SC-branches will only be reduced to two critical cliques. The connected components of the graph $T \setminus (A \cup A')$ are paths, which correspond to subgraphs of $2$-SC-branches, and we will bound their length by $k+3$.
We know that there are at most $4k$ such connected components, thus there is $O(k^2)$ critical cliques in $T$. Finally, the connected components of the graph $\C(H)\setminus V(T)$ correspond to $1$-SC-branches or sets of connected $1$-SC-branches. The latter case will require a specific reduction rule to be dealt with. This implies that every critical clique or maximal clique of $T$ can have a linear number of $1$-SC-branches of $\C(H)\setminus V(T)$ adjacent to it. Altogether, the graph $\C(H)$ contains $O(k^3)$ critical cliques, and each critical clique is of size at most $k+1$, hence the graph $G$ contains $O(k^4)$ vertices. \\

From now on, we always assume that the 
considered optimal editions (resp. completions, deletions) satisfy \cref{lem:homogen-compl-here}.

\subsection{Classical reduction rules}
First, we give classical reduction rules when dealing with modification problems. The first rule is safe for any target graph class hereditary and closed under disjoint union. The second 
one comes from the fact that strictly chordal graphs are hereditary and closed under true twin addition, combined with \cref{lem:homogen-compl-here}. 

\begin{polyrule} \label{rule:compSC}
    Let $C$ be a \bdsc{} connected component of $G$. 
    Remove $V(C)$ from $G$.
\end{polyrule}

\begin{polyrule}\label{rule:borneCC}
Let $K \subseteq V$ be a set of true twins of $G$ such that $|K| > k+1$. Remove $|K|-(k+1)$ arbitrary vertices in $K$ from $G$. 
\end{polyrule}

\begin{lemma}[Folklore, \cite{BPP10}]
\label{lem:safecc}
\cref{rule:compSC,rule:borneCC} are safe and can be applied in polynomial time. 
\end{lemma}

\subsection{Reducing Strictly Chordal branches}

We now consider the main structure of our kernelization algorithm, namely \emph{SC-branches}.

\begin{definition}[SC-branch]
Let $G=(V,E)$ be a graph. We say that a subgraph $B$ of $G$ is a \emph{SC-branch} if $V(B)$ is a union of critical cliques $K_1,\dots ,K_r \in \mathcal{K}(G)$ such that the subgraph of $\C(G)$ induced by $K_1,\dots ,K_r$ is a connected block graph.
\end{definition}

We emphasize that our definition of a SC-branch $B$ is stronger than simply requiring $B$ to be an induced \bdsc\ graph. For example, if $G$ is the dart graph, the subgraph obtained by removing the pendant vertex is \bdsc, but it is not a SC-branch because the corresponding critical cliques do not form a block graph in $\C(G)$. 
Let $B$ be a SC-branch of graph $G$ and let $K_1,\dots ,K_r$ be the critical cliques of $G$ contained in $V(B)$. We say that $K_i$ is an \emph{attachment point} of $B$ if $N_G(K_i) \setminus V(B) \neq \emptyset$.
A SC-branch $B$ is a $p$-SC-branch if it has exactly $p$ attachment points. We denote $B^R$ the induced subgraph of $B$ in which all attachment points have been removed. 

We first give structural lemmata on SC-branches that will be helpful to prove the safeness of our rules.

\begin{lemma}%
\label{lem:1-block}
Let $G = (V,E)$ be a graph and $B$ a SC-branch of $G$. For any attachment point $P$ of $B$, let $B' = B\setminus P$, consider the connected components $G_1, G_2, \dots, G_r$ of $B'$ and let $Q_i = N_B(P) \cap V(G_i)$. For every $i,\ 1 \leq i \leq r$, $Q_i$ is a critical clique, a maximal clique or intersects exactly one maximal clique of $G_i$.

\end{lemma}

\begin{proof}

First, we show that all sets $Q_i$ are cliques. Suppose that $Q_i$ is not a clique for some $1 \leq i \leq r$. Let $x$ and $y$ be non-adjacent vertices of $Q_i$ and $z\in P$. Since $G_i$ is connected, take a shortest path $\pi$ between $x$ and $y$ in $G_i$. The subgraph induced by the vertices $\{x,y,z\}$ and those of $\pi$ contains either a cycle of length at least $4$ if $z$ is not adjacent to any inner vertex of $\pi$, which is a contradiction since $B$ is a block graph, or a diamond with $z$ being one of its vertices of degree $3$. In the latter case, since $z$ is not in the same critical clique of $\C(G)$  as its true twin in the diamond, the critical cliques of $\C(G)$ that contain some vertices of this diamond also form a diamond in $\C(G)$.
Such a diamond is formed by critical cliques of $G$ contained in $B$, contradicting the definition of a SC-branch. 
In all cases we have a contradiction, it remains that $Q_i$ is a clique in $G$.\\  
Now suppose that $Q_i$, $1 \leq i \leq r$ is neither a critical clique nor a maximal clique of $G_i$ and intersects at least two maximal cliques of $G_i$. Suppose first that $Q_i$ is not a module, then there exist $u,v\in Q_i$ and $x \in V(G_i) \setminus Q_i$ adjacent to only $u$ or $v$. Since $Q_i$ is not a maximal clique of $G_i$, there exists $y \in V(G_i) \setminus Q_i$ adjacent to $u$ and $v$. Together with $w\in P$, the vertex set $\{u,v,w,x,y\}$ induces in $B$ a gem if $x$ and $y$ are adjacent and a dart otherwise, leading to a contradiction.
If $Q_i$ is a module, since it is not a critical clique (\ie a maximal clique module) there exist $u\in Q_i$ and $v\in  V(G_i) \setminus Q_i$ that lie in the same critical clique of $G_i$. Since $Q_i$ intersects two maximal cliques of $G_i$, there exist two non-adjacent vertices $x,y \in V(G_i) \setminus Q_i$ that are adjacent to both $u$ and $v$. Together with $w\in P$, the vertex set $\{u,v,w,x,y\}$ induces a dart in $B$, leading to a contradiction. It remains that $Q_i$ is a critical clique, a maximal clique or intersects exactly one maximal clique of $G_i$. 
\end{proof}

\begin{lemma} \label{lem:1-block_edition}
Let $G=(V,E)$ be a graph and $B$ a SC-branch of $G$. Let $F$ be an optimal edition of $G$ that respects \cref{lem:homogen-compl-here}, and $H = G \triangle F$. For any attachment point $P$ of $B$, let $C$ be the critical clique of $H$ which contains $P$ and $C' = C \setminus V(B^R)$. If $N_H(C') \cap V(B^R) \neq \emptyset$, then $C'$ is a critical clique or intersects exactly one maximal clique of $H' = H\setminus V(B^R)$.
\end{lemma}

\begin{proof}
Suppose that $C'$ is not a critical clique of $H'$ and intersects at least two maximal cliques of $H'$. Since $C' \subseteq C$ is a module but not a critical clique, it is not a maximal clique module and thus there exist $u \in C'$ and $v \notin C'$ that lie in the same critical clique in $H'$. 
Notice that there must exists a vertex $w \in N_H(C') \cap V(B^R)$ 
adjacent to exactly one of $u$ and $v$ in $H$, since otherwise we would have 
$N_H(u)  \cap V(B) = N_H(v) \cap V(B)$ and $N_{H'}(u) = N_{H'}(v)$, meaning that $C$ is 
not a critical clique of $H$, a contradiction. 
Since $C'$ intersects two maximal cliques of $H'$ there exist two non-adjacent vertices $x,y \in V(H')$ that are adjacent to both $u$ and $v$. The vertex set $\{u,v,w,x,y\}$ induces a dart in $H$ if $w$ is not adjacent to neither $x$ nor $y$, a gem if it is adjacent to one of them, and if $w$ is adjacent to both $x$ and $y$, $\{w,x,v,y\}$ induces a $C_4$ in $H$. It remains that $C'$ is a critical clique of $H'$ or $C'$ is included in exactly one maximal clique of $H'$ if $N_H(C') \cap V(B^R) \neq \emptyset$.
\end{proof}

In the next result we show that there always exists an optimal edition that affects only the attachment point of a 1-SC-branch and its neighborhood. This will allow us to define our first reduction rule on SC-branches.

\begin{lemma} \label{lem:1-SC-branchE}
Let $G=(V,E)$ be a graph and $B$ a $1$-SC-branch of $G$ with attachment point $P$. There exists an optimal edition $F$ of $G$ such that:
\begin{itemize}
    \item The set of vertices of $B$ affected by $F$ is included in $P \cup N_B(P)$.
    \item In $H = G \triangle F$ the vertices of $N_B(P)$ are all adjacent to the same vertices of $V(G) \setminus V(B^R)$.
\end{itemize}
\end{lemma}

\begin{proof}
Let $F$ be an optimal edition of $G$ and let $H= G\triangle F$. Recall that we assume that $F$ satisfies \cref{lem:homogen-compl-here}. 
Let $G_1, G_2, \dots, G_r$ be the connected components of $B^R$ and let $Q_i = N_B(P) \cap V(G_i)$. By \cref{lem:1-block}, for every $i,\ 1 \leq i \leq r$, $Q_i$ is a critical clique, a maximal clique or intersects exactly one maximal clique of $G_i$. Moreover, the graphs $H' = H \setminus V(B^R)$ and $G_1, \dots , G_r $ are strictly chordal by heredity. 
Among the vertices of $Q_1,\dots ,Q_r$, choose one that is in the least number of pairs $(x,b)\in F$ with $b \in V(G) \setminus V(B^R)$, say $x$, and let $N = N_H(x) \setminus V(B^R)$. Observe that since $P$ is a critical clique, $P \cap N = \emptyset$ or $P \subseteq N$. If $P \subseteq N$, let $C$ be the critical clique of $H$ which contains $P$, take $C' = C\setminus V(B^R)$ and observe that $P \subseteq C' \subseteq N$. By \cref{lem:1-block_edition}, $C'$ is a critical clique or intersects exactly one maximal clique of $H'$. By \cref{lem:constructionSC} and \cref{obs:SCmulticomp}, the following graph is is strictly chordal: $$H^* = (\bigcup_{1 \leq i \leq r}G_i, \bigcup_{1 \leq i \leq r}Q_i) \otimes (H', C')$$ If $P \cap N = \emptyset$, let $H^*$ be the strictly chordal graph corresponding to the disjoint union of $H'$ and the graphs $G_1, \dots, G_r$. Let $F^*$ be such that $H^* = G \triangle F^*$, in both cases by construction $|F^*| \leq |F|$ and the desired properties are verified. 
\end{proof}

\begin{figure}
    \centering
    \includegraphics[width=7.5cm]{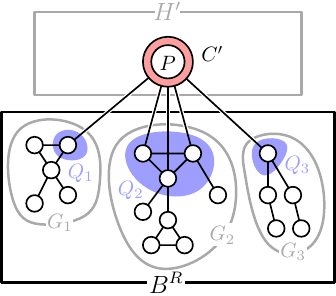}
    \caption{Illustration of the proof of \cref{lem:1-SC-branchE} for $r=3$. The graph $H^* = (G_1 \cup G_2 \cup G_3, Q_1 \cup Q_2 \cup Q_3) \otimes (H', C')$  is strictly chordal by \cref{lem:constructionSC}.
    }
    \label{fig:lemme-1-SC-branchC}
\end{figure}

\begin{polyrule} \label{rule:1-SC-branchE}
Let $(G=(V,E),k)$ be an instance of \SCE{}. If $G$ contains a $1$-SC-branch $B$ with attachment point $P$, then remove from $G$ the vertices of $V(B^R)$ and add a clique $K$ of size $min\{|N_B(P)|, k+1\}$ adjacent to $P$.
\end{polyrule}

\begin{lemma}
\cref{rule:1-SC-branchE} is safe.
\end{lemma}

\begin{proof}
Let $G'$ be the graph obtained from an application of \cref{rule:1-SC-branchE} on $1$-SC-branch $B$ of graph $G$ with attachment point $P$. We denote by $K$ be the clique that replaced $B$. 
Let $F$ be a $k$-edition of $G$ and $H = G \triangle F$. By \cref{lem:1-SC-branchE} we can assume that $F$ only affects vertices of $P \cup N_B(P)$ and all vertices of $N_B(P)$ are adjacent to the same vertices of $V(G) \setminus V(B^R)$. This implies that if $|N_B(P)| > k$ no $k$-edition can affect these vertices. Let $C$ be the critical clique of $H$ which contains $P$ and $C' = C\setminus V(B^R)$. By \cref{lem:1-block_edition}, $C'$ is a critical clique or intersects exactly one maximal clique of $H_a = H \setminus V(B^R)$. If $N_B(P)$ is adjacent to $C'$ in $H$, consider the graph $H' = (H_a, C') \otimes (K,K)$. Since $K$ is clique, by \cref{lem:constructionSC}, $H'$ is strictly chordal. Else, consider the strictly chordal graph $H'$ corresponding to the disjoint union of $H_a$ and $K$. Let $F'$ be the edition such that $H' = G' \triangle F'$. Since $|K| = min\{|N_B(P)|,k+1\}$, by construction we have $|F| = |F'|$, hence $F'$ is a $k$-edition of $G'$.

The other way is almost symmetrical, but we need to be careful that in this case, $B^R$ may contain several connected components. 
Let $F'$ be a $k$-edition of $G'$ and $H' = G' \triangle F'$. By \cref{lem:1-SC-branchE} we can assume that $F'$ only affects vertices of $P \cup K$ and all vertices of $K$ are adjacent to the same vertices of $V(G) \setminus K$. This implies that, if $|K| > k$, no $k$-edition can affect these vertices. Let $C$ be the critical clique of $H'$ which contains $P$ and $C' = C\setminus V(B^R)$. By \cref{lem:1-block_edition}, $C'$ is a critical clique or intersects exactly one maximal clique of $H_a' = H' \setminus K$. 
Let $G_1, G_2, \dots, G_r$ be the connected components of $B^R$ and let $Q_i = N_B(P) \cap V(G_i)$. By \cref{lem:1-block}, for every $i,\ 1 \leq i \leq r$, $Q_i$ is a critical clique, a maximal clique or intersects exactly one maximal clique of $G_i$.

If the clique $K$ is adjacent to $C'$ in $H'$, consider the graph $H = (H_a', C') \otimes (\bigcup_{1 \leq i \leq r}G_i, \bigcup_{1 \leq i \leq r}Q_i)$. By \cref{lem:constructionSC} and \cref{obs:SCmulticomp}, $H'$ is strictly chordal. Else, consider the strictly chordal 
graph $H$ corresponding to the disjoint union of $H_a'$ and $G_i$ for $1 \leq i \leq r $. Let $F$ be the edition such that $H = G \triangle F$, since $|K| = min\{|N_B(P)|,k+1\}$, by construction we have $|F| = |F'|$, hence $F$ is a $k$-edition of $G$.
\end{proof}

We now show a result that will allow us to define a rule bounding the number of $1$-SC-branches sharing the same neighborhood. In particular, if enough $1$-SC-branches share a neighborhood $N$, they force any edition of size at most $k$ to transform $N$ into a clique.

\begin{lemma}  \label{lem:voisinage-communBBE}
Let $(G=(V,E),k)$ be a yes-instance of \SCE{} reduced by \cref{rule:borneCC}. Let $B_1,\dots , B_l$ be disjoint $1$-SC-branches of $G$ with attachment points $P_1,\dots ,P_l$ which have the same neighborhood $N$ in $G\setminus \bigcup_{i=1}^l V(B_i)$ and form a disjoint union of cliques $Q_1,\dots Q_r$ in $G[P_1\cup\dots \cup P_l]$. If $\Sigma_{i=1}^l |P_i| > 2k+1$ then for every $k$-edition $F$ of $G$:

\begin{enumerate}[(i)]
    \item\label{vcBBE:i} If $r=1$, $N$ is a clique of $H=G \triangle F$ and $N$ is a critical clique, a maximal clique or intersects exactly one maximal clique of $H' = H\setminus \bigcup_{i=1}^l V(B_i)$,
    \item\label{vcBBE:ii} If $r>1$  and $(\Sigma_{i=1}^l |P_i|) - \max_{1\leq j \leq r}\{|Q_j|\} > k$ then $N$ is a critical clique of $H=G \triangle F$ and $N$ is a critical clique or intersects exactly one maximal clique of $H' = H\setminus \bigcup_{i=1}^l V(B_i)$. 
\end{enumerate}
\end{lemma}

\begin{proof} We prove the result by considering each item in separate cases. 
\paragraph{Case~(\ref{vcBBE:i}) $r = 1$:} We first prove that $N$ has to be clique in $H$. Suppose that $N$ is not a clique in $H$, then there exist two vertices $x,y \in N$ such that $xy \notin E(H)$. The attachment points are critical cliques, so by \cref{rule:borneCC}, $|P_i| \leq k+1, i\in \{1, \dots , l\}$.  We can observe that  there is at most one critical clique $P_j$ such that $N_{B_j}(P_j) = \emptyset$. Indeed, since $r=1$ all attachment points lie in the same clique in $G[P_1\cup\dots \cup P_l]$, we have that $N_G[P_i]\setminus N_{B_i}(P_i) = N \cup P_1\cup\dots \cup P_l$, hence there is only one such $P_j$. It follows that $|Q_1 \setminus P_j|  > k$, thus there exists a vertex $u$ in some $P_t, t\neq j$ such that in $H$, $u$ is adjacent to a vertex $z \in N_{B_t}(P_t)$ and the vertices $x$ and $y$. Since $|Q_1\setminus P_t|  > k$, there exists a vertex $v \in Q_1\setminus P_t$ such that there is no pair of $F$ containing both $v$ and a vertex of $\{x,y,z,u\}$. The set $\{x,y,u,v,z\}$ induces a dart if neither $x$ nor $y$ is adjacent to $z$, a gem if one of them is adjacent to $z$, or if both of them are adjacent to $z$, $\{x, y,  z, v\}$ induces a $C_4$ in $H$, a contradiction. It follows that $N$ has to be a clique.

We now show that $N$ is a critical clique, a maximal clique or intersects exactly one maximal clique of $H'$. Suppose that $N$ is neither a critical clique nor a maximal clique of $H'$ and intersects at least two maximal cliques of $H'$. Suppose first that $N$ is not a module, there exist $u,v\in N$, $x \in V(H') \setminus N$ adjacent to only $u$ or $v$. Since $N$ is not a maximal clique of $H'$, there exists $y \in V(H') \setminus N$ adjacent to $u$ and $v$. Together with a vertex $w\in \bigcup_{i=1}^l P_i$ not affected by $F$, the vertex set $\{u,v,w,x,y\}$ induces in $H$ a gem if $x$ and $y$ are adjacent and a dart else, leading to a contradiction. If $N$ is a module and is not a critical clique, $N$ is not a maximal clique module, thus there exist $u \in N$, a vertex $v$ in the same module as $u$ in $H'$ and a vertex $w\in \bigcup_{i=1}^l P_i$ not affected by $F$ such that $vw \notin E(H)$. Since $N$ intersects two maximal cliques of $H'$ there exist two non-adjacent vertices $x,y \in V(H')$ that are adjacent to both $u$ and $v$. The vertex set $\{u,v,w,x,y\}$ induces a dart in $H$ since $w$ is only adjacent to $u$. It remains that if $r=1$, $N$ is a critical clique, a maximal clique or is included in exactly one maximal clique of $H'$.

\paragraph{Case~(\ref{vcBBE:ii}) $r > 1$ and $(\Sigma_{i=1}^l |P_i|) - \max_{1\leq j \leq r}\{|Q_j|\} > k$:} We claim that $N$ is a clique in $H$. Suppose again for a contradiction that this is not the case, then there exist two vertices $x,y \in N$ such that $xy \notin E(H)$. For any pair of vertices $u_i \in Q_i, u_j \in Q_j, i\neq j$, the set $\{x,u_i,y,u_j\}$ induces a $C_4$ in $H$ if neither $u_j$ nor $u_i$ is affected by $F$. It follows that the vertices of $Q_i$ or $Q_j$ must be affected by $F$. This implies that at most one $Q_i$ is unaffected by $F$, and since $(\Sigma_{i=1}^l |P_i|) - \max_{1\leq j \leq r}\{|Q_j|\} > k$ we have $|F| > k$, a contradiction. Hence $N$ has to be a clique in $H$. 

Next we show that $N$ is a critical clique of $H$. Suppose that $N$ is not a module in $H$, then there exist $ x,y \in N$ and $ z \in V(G) \setminus N$ adjacent to only one of the vertices $x$ or $y$, say w.l.o.g. $x$. For any pair of vertices $u_i \in Q_i, u_j \in Q_j, i\neq j$, if neither $u_i$ nor $u_i$ is affected by $F$, then at least one of them is not adjacent to $z$ and the set $\{x,y,u_i,u_j,z\}$ induces a dart if neither $u_i$ nor $u_j$ is adjacent to $z$ or a gem if one of them is adjacent to $z$. This implies that at most one $Q_i$ is unaffected by $F$, as before we have  $|F| > k$, a contradiction. Hence $N$ has to be a module in $H$. 
If $N$ is not a critical clique of $H$, then it is strictly contained in some critical clique $N'$. A vertex $x \in N' \setminus N$ must have the same closed neighborhood than $N$ in $H$. If $x$ is not in some $P_i$ or $N_{B_i}(P_i)$, then $F$ must affect all vertices of $P_1 \cup \dots \cup P_l$, implying $|F| > k$, a contradiction. 
If $x$ is in some $P_i$ or $N_{B_i}(P_i)$, then  $F$ must affect all vertices of every $Q_j$ such that $P_i\nsubseteq Q_j$. Since $(\Sigma_{i=1}^l |P_i|) - \max_{1\leq j \leq r}\{|Q_j|\} > k$, this implies $|F| > k$, a contradiction. It remains that $N$ is a critical clique in $H$. 

Finally, since $N$ is a module in $H'$, if it is not a critical clique, there exist $u \in N$, a vertex $v$ in the same critical clique as $u$ in $H'$ and a vertex $w\in \bigcup_{i=1}^l P_i$ not affected by $F$ such that $vw \notin E(H)$. If $N$ intersects at least two maximal cliques of $H'$ there exist two non-adjacent vertices $x,y \in V(H')$ that are adjacent to both $u$ and $v$. The vertex set $\{u,v,w,x,y\}$ induces a dart in $H$ since $w$ is only adjacent to $u$. It remains that $N$ is a critical clique or intersects exactly one maximal clique of $H'$.
\end{proof}

\begin{polyrule}\label{rule:voisinage_communBBE}
Let $(G=(V,E),k)$ be an instance of \SCE{}. Let $B_1,\dots , B_l$ be disjoint $1$-SC-branches of $G$ with attachment points $P_1,\dots ,P_l$ which have the same neighborhood $N$ in $G\setminus \bigcup_{i=1}^l V(B_i)$ and form a disjoint union of cliques $Q_1,\dots Q_r$ in $G[P_1\cup\dots \cup P_l]$. Assume that $\Sigma_{i=1}^l |P_i| > 2k+1$, then :
\begin{itemize}
    \item If $r=1$, remove the vertices $\bigcup_{i=1}^l V(B_i)$, add two adjacent cliques $K$ and $K'$ of size $k+1$ with neighborhood $N$ and a vertex $u_K$ adjacent to every vertex of $K$.
    \item If $r>1$ and $(\Sigma_{i=1}^l |P_i|) - \max_{1\leq j \leq r}\{|Q_j|\} > k$, remove the vertices $\bigcup_{i=1}^l V(B_i)$ and add two non-adjacent cliques of size $k+1$ with neighborhood $N$.
\end{itemize}
\end{polyrule}

\begin{lemma}
\cref{rule:voisinage_communBBE} is safe.
\end{lemma}

\begin{proof}
Let $G'$ be the graph obtained after the application of \cref{rule:voisinage_communBBE} on $B_1,\dots , B_l$, and $S$ the new vertices introduced by this rule. Notice that the subgraph $G_S' = G'[S]$ is strictly chordal and form two $1$-SC-branches that satisfy the conditions of \cref{lem:voisinage-communBBE}. First, let $F$ be a $k$-edition of $G$, $H = G \triangle F$ and $H_a = H\setminus \bigcup_{i=1}^l V(B_i)$. By \cref{lem:voisinage-communBBE} $N$ is a critical clique, a maximal clique (unless $r>1$) or intersects exactly one maximal clique of $H_a$. Therefore, by \cref{lem:constructionSC} and \cref{obs:SCmulticomp} the graph $H' = (H_a, N) \otimes (G_S', K\cup K')$ is strictly chordal and the edition $F'$ such that $H' = G' \triangle F'$ is of size at most $|F| \leq k$. 

Conversely, let $F'$ be a $k$-edition of $G'$, $H' = G' \triangle F'$ and $H_a' = H'\setminus S$. Notice that the subgraph $G_b = G[\bigcup_{i=1}^l V(B_i)]$ is strictly chordal. By \cref{lem:voisinage-communBBE} $N$ is a critical clique, a maximal clique (unless $r>1$) or intersects exactly one maximal clique of $H_a'$. Let $Q_1, \dots , Q_r$ the cliques of $G_b$ adjacent to $N$, using arguments similar to the ones of \cref{lem:1-block}, it is easy to see that each $Q_i$ is critical clique, a maximal clique or intersects exactly one maximal clique of $G_b$. Therefore, by \cref{lem:constructionSC} and \cref{obs:SCmulticomp} the graph $H = (H_a', N) \otimes (G_b, Q_1 \cup \dots \cup Q_r)$ is strictly chordal and the edition $F$ such that $H = G \triangle F$ is of size at most $|F'| \leq k$. 
\end{proof}

\subsection{Reducing the 2-SC-branches}

Let $B$ be a $2$-SC-branch of a graph $G$ reduced by \cref{rule:1-SC-branchE},  with attachment points $P_1$ and $P_2$. We say that $B$ is \emph{clean} if $B^R$ is connected, and that the \emph{length} of a clean $2$-SC-branch is the length of a shortest path between the two attachment points in $\C(B)$. We say that the \emph{min-cut} of $B$ is a set of edges $M \subseteq E(B)$ of minimum size such that $P_1$ and $P_2$ are not in the same connected component of $B-M$ and $mc(B)$ is the size of a min-cut of $B$.
We can observe that a min-cut of $B$ contains the edges between a pair of consecutive critical cliques of a shortest path between $P_1$ and $P_2$ and the edges between one of these critical cliques and the critical cliques they have in their common neighborhood. The following observation will be useful to reduce $2$-SC-branches. 

\begin{observation}\label{obs:2ed}
Let $F$ be an optimal edition of $G$ and $F_1 \subseteq F$. If $F_2$ is an optimal edition of the graph $G\triangle F_1$ , then $F_1 \cup F_2$ is an optimal edition of $G$.
\end{observation}

\begin{lemma} \label{lem:2-SC-branchE}
Let $(G=(V,E),k)$ be a yes-instance of \SCE{} and $B$ a clean $2$-SC-branch of length at least $k+3$ with attachment points $P_1,P_2$. There exists an optimal edition $F$ of $G$ such that:
\begin{itemize}
    \item  If $P_1$ and $P_2$  are in the same connected component of $(G\triangle F)\setminus V(B^R)$, then $F$ contains an edge-cut of $B$. 
    This edge-cut is $(P_1 \times N_B(P_1))$ or $(P_2 \times N_B(P_2))$ or a min-cut of $B$.
    \item In each case, the other vertices of $B$ affected by $F$ are included in $P_1 \cup N_B(P_1) \cup P_2 \cup N_B(P_2)$. 
    \item In $G\triangle F$ the vertices of $N_B(P_1)$ (resp. $N_B(P_2)$) are all adjacent to the same vertices of $V(G) \setminus V(B^R) $.

\end{itemize}
\end{lemma}

\begin{proof} Let $F$ be an optimal edition of $G$, $H = G \triangle F$. Let $C_1$ (resp. $C_2$) be the critical clique of $H$ which contains $P_1$ (resp. $P_2$). Let $H' = H\setminus V(B^R)$, $C_1' = C_1\setminus V(B^R)$ and $C_2' = C_2\setminus V(B^R)$.  

We first consider the case where $P_1$ and $P_2$ are in different connected component of $H'$. Let $H'_1$ and $H'_2$ be the two connected components of $H'$ containing $P_1$ and $P_2$, respectively, and $H'_3$ the remaining of $H'$.
For $i\in \{1,2\}$, \cref{lem:1-block} implies that $N_B(P_i)$ is a clique in $G$ since $B^R$ is connected, and $N_B(P_i)$ is a critical clique, a maximal clique or intersects exactly one maximal clique of $B^R$. By \cref{lem:1-block_edition}, $C_i'$ is a critical clique or intersects exactly one maximal clique of the connected component $H'_i$. 

By \cref{lem:constructionSC}, the graph $H^*$ corresponding to the disjoint union of $$((B^R, N_B(P_1)) \otimes (H'_1, C_1')), N_B(P_2)) \otimes (H'_2, C_2')$$ and $H'_3$ is strictly chordal (see \cref{fig:2-SC-branch}) and the edition $F^*$ such that $H^* = G \triangle F^*$ verifies the desired properties. Since  $|F^*| \leq |F|$ by construction, we conclude that $F^*$ is optimal. 

\begin{figure}[h]
    \centering
    \includegraphics[width=11cm]{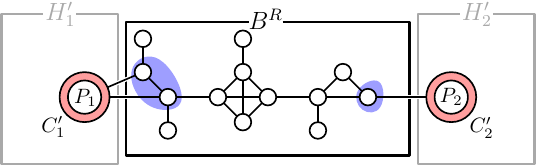}
    \caption{Illustration of the case where $B\triangle F$ is connected in \cref{lem:2-SC-branchE}.}
    \label{fig:2-SC-branch}
\end{figure}

Assume now that $P_1$ and $P_2$ are in the same connected component in $H'$ and let $\pi_{H'}$ denote a shortest path between $c_1 \in P_1$ and $c_2 \in P_2$. Note that $\pi_{H'}$ may be of length $1$ if $c_1$ and $c_2$ are adjacent in $H'$. We first consider the case where there is a path $\pi_B=\{c_1=u_1,u_2,\dots,u_r=c_2\}$ (not necessarily induced) from $c_1$ to $c_2$ in $B$ that still exists in $B \triangle F$. Observe that such this path $\pi_B$ is of length at least $k+3$ since $B$ is of length at least $k+3$.
There is no edge between the two paths in $G$, so $H[V(\pi_B) \cup V(\pi_H)]$ admits at least one induced cycle of length at least 4, which is a contradiction. We conclude that if $P_1$ and $P_2$  are in the same connected component of $H\setminus V(B^R)$, then $F$ contains an edge-cut of $B$. Else there would be a path $\pi_B$ from $c_1$ to $c_2$ in $B$ that still exists in $B\triangle F$.

Hence we can assume that $F$ contains an edge-cut of $B$.
First, if $F_1 = (P_1 \times N_B(P_1)) \subseteq F$, consider the graph $G_1 = G \triangle F_1$. Observe that $B_1 = B \setminus P_1$ is a $1$-SC-branch of $G_1$ with attachment point $P_2$. By \cref{lem:1-SC-branchE} there exists an optimal edition $F_2$ of $G_1$ where the vertices of $B_1$ affected by $F_2$ are in $P_2 \cup N_{B_1}(P_2)$. By \cref{obs:2ed}, $F_1 \cup F_2$ is an optimal edition and it verifies the desired properties. The same arguments can be used if $(P_2 \times N_B(P_2)) \subseteq F$. 

Now, if $F$ contains neither $(P_1 \times N_B(P_1))$ nor $(P_2 \times N_B(P_2))$, let $F_1 = F \cap (V(B^R) \times V(B^R))$.
There are two components in $B \triangle F_1$, say $B_1$ the one containing $P_1$ and $B_2$ one containing $P_2$. These connected components are $1$-SC-branches of $G_1 = G \triangle F_1$ with attachment points $P_1$ and $P_2$. As before \cref{lem:1-SC-branchE} applies on the $1$-SC-branches $B_1$ and $B_2$, thus there is an optimal edition $F_2$ of $G_1$ where the only vertices of $B_1$ and $B_2$ affected are included in $P_1 \cup N_B(P_1)$ and $P_2 \cup N_B(P_2)$. If $F_1$ is a min-cut of $B$, the edition $F_1 \cup F_2$ is an optimal edition of $G$ that verifies the desired properties, else, replacing $F_1$ by a min-cut of $B$ results in a smaller edition that verifies the properties. 
\end{proof}

\begin{polyrule} \label{rule:2-SC-branchE}
Let $(G=(V,E),k)$ be an instance of \SCE{} and $B$ a clean $2$-SC-branch of $G$ of length at least $k+3$ with attachment points $P_1, P_2$. Remove the vertices $V(B^R)$, and add the following path of $k+2$ cliques:  
\begin{center}
$K_{min\{|N_B(P_1)|, k+1\}} - K_{k+1} - K_1 - K_{|mc(B)|} - K_{k+1}^1 - \dots - K_{k+1}^{k-3} - K_{min\{|N_B(P_2)|, k+1\}}$
\end{center} 
\noindent where $K_n$ is the clique of size $n$ and  $K_{min\{|N_B(P_1)|, k+1\}}$ (resp. $K_{min\{|N_B(P_2)|, k+1\}}$) is adjacent to $P_1$ (resp. $P_2$).
\end{polyrule}

\begin{lemma} \label{lem:2-SC-branchE-saferule}
\cref{rule:2-SC-branchE} is safe.
\end{lemma}

\begin{proof}
Let $G'$ the graph obtained from an application of \cref{rule:2-SC-branchE} on the $2$-SC-branch $B$ of graph $G$ with attachment points $P_1,P_2$ and $S$ the vertices introduced by this rule. Observe that $B' = G'[S \cup P_1 \cup P_2]$ is a clean $2$-SC-branch of $G'$ of length $k+3$ and $mc(B') = min \{k+1, mc(B)\}$. 

Let $F$ be a $k$-edition of $G$ that satisfies \cref{lem:2-SC-branchE} and $H = G \triangle F$. Observe that if $mc(B) \geq k+1$, then $F$ does not affect vertices of $V(B)$ other than $P_1 \cup N_B(P_1) \cup P_2 \cup N_B(P_2)$ and if $|N_B(P_i)| > k, i\in \{1,2\}$ $F$ does not affect the vertices of $N_B(P_i)$. 
Let $C_i$ be the critical clique of $H$ which contains $P_i$ and $C_i' = C_i\setminus B^R$. 
By \cref{lem:1-block_edition}, $C'_i$ is a critical clique or intersects exactly one maximal clique of $H_a = H \setminus V(B^R)$. Observe moreover that $N_{B'}(P_i)$ is critical clique of $B'^R$. 

We now construct a $k$-edition $F'$ of $G'$. We consider first the tree cases where $F$ contains an edge-cut of $B$.
If w.l.o.g. $(P_1 \times N_B(P_1)) \subseteq F$ and $(P_2 \times N_B(P_2)) \cap F = \emptyset$. Observe that $(C_1' \times N_{B'}(P_1) \subseteq F$ since $C_1$ is a critical clique of $H$. By \cref{lem:constructionSC} the graph $H' = ((H_a, C_1') \otimes (B'^R,N_{B'}(P_1))$ is strictly chordal.
If $(P_1 \times N_B(P_1)) \subseteq F$ and $(P_2 \times N_B(P_2)) \subseteq F$ then $H'$ the graph resulting of the disjoint union of $H_a$ and $B'^R$ is strictly chordal.
If $F$ contains a min-cut of $B$, then consider $B_i', i\in \{1,2\}$ the connected component of $B' \triangle F_1'$ containing $P_i$, where $F_1'$ is a min-cut of $B'$. By \cref{lem:constructionSC} the graph $H' = ((H_a, C_1') \otimes (B_1',N_{B'}(P_1)), C_2') \otimes (B_2',N_{B'}(P_2))$ is strictly chordal. 

Otherwise, $P_1$ and $P_2$ are not in the same connected component of $H_a$ by \cref{lem:2-SC-branchE}. Let $H_1$ and $H_2$ be these components and $H_3$ be the remaining of $H_a$. By \cref{lem:constructionSC}, the graph $H'$ corresponding to the disjoint union of $((B'^R, N_{B'}(P_1)) \otimes (H_1, C_1'), N_{B'}(P_2)) \otimes (H_2, C_2')$ and $H_3$ is strictly chordal. 

In any of these four cases, let $F'$ be the edition such that $H'  = G' \triangle F'$, by construction $|F'| \leq |F|$, therefore $F'$ is a $k$-edition of $G'$. 

The other way is symmetrical.
\end{proof}

\subsection{Bounding the size of the kernel}
\label{sec:size}

We first state that reduction rules involving SC-branches can be applied in polynomial time.

\begin{lemma}\label{lem:detect}
    Given an instance $(G = (V,E),k)$ of \SCE{}, \cref{rule:1-SC-branchE,rule:voisinage_communBBE,rule:2-SC-branchE} can be applied in polynomial time. 
\end{lemma}

\begin{proof}
    We rely on the linear time computation of the critical clique graph $\C(G)$ of $G$~\cite{PDS09} and the linear time recognition of block graphs. 
    We show that we can enumerate all $1$-SC-branches and $2$-SC-branches in polynomial time. Since an attachment point is by definition 
    a critical clique of $G$, one can detect $1$-SC-branches by removing a
    critical clique of $\C(G)$ and looking among the remaining connected components
    those that induce a connected block graph together with $P$ (in $\C(G)$). 
    Considering a maximal set of such components together with $P$ gives a $1$-SC-branch $B$. 
    We proceed similarly to detect clean $2$-SC-branches by removing a pair of critical cliques $P_1, P_2$ of $\C(G)$, and looking among the remaining connected components those that induce a connected block graph (in $\C(G)$) together with $P_1$ and $P_2$. Such components together with $P_1$ and $P_2$ gives a clean $2$-SC-branch $B$. 
    Since there is $O(|V(G)|)$ critical cliques, the result follows. 
\end{proof}



\THMSCE*

\begin{proof} 
Let $(G=(V,E),k)$ be a reduced yes-instance of \SCE, $F$ a $k$-edition of $G$ and $H=G \triangle F$. We assume that $G$ is connected, the following arguments can be easily adapted if this is not the case by summing over all connected components of $G$. The graph $H$ is strictly chordal, thus the graph $\C(H)$ is a block graph. We first show that $\C(H)$ has  $O(k^3)$ vertices. 
We say that a critical clique of $\C(H)$ is affected if it contains a vertex affected by $F$. Let $A$ be the set of affected critical cliques of $\C(H)$. Since $|F| \leq k$, we have $|A| \leq 2k$. Let $T_A$ be the minimal connected induced subgraph of $H$ that spans all critical cliques of $A$, and $\sd{T_A}$ the set of vertices of degree at least $3$ in $T_A$ (see \cref{fig:ex_graphe_reduit}). Recall that such a subgraph is unique 
by \cref{lem:bloc-span}. 

We first show that $|V(T_A)|$ contains $O(k^2)$ critical cliques. By \cref{lem:bloc-span}, $|\sd{T_A} \setminus A| \leq 3\cdot |A| \leq 6k$. The connected components of the graph $T_A \setminus (A \cup \sd{T_A})$ are paths since every vertex is of degree at most $2$ and by \cref{lem:bloc-span} there are at most $4k$ such paths. Let $R$ be one of these paths, it is composed of unaffected critical cliques, thus there exists a clean $2$-SC-branch $B$ of $G$ that contains $R$. Moreover, the endpoints of $R$ are the attachment points of $B$, hence $B$ have been reduced by \cref{rule:2-SC-branchE}. Thus $R$ is of length at most $k+3$. It remains that $|V(T_A)| \leqslant 2k + 6k + (k+3) \cdot 4k = O(k^2)$.

We will now show that $\C(H) \setminus V(T_A)$ contains $O(k^3)$ critical cliques. First observe that each connected component of $\C(H) \setminus V(T_A)$ is adjacent to some vertices of $T_A$ since the graph is reduced by \cref{rule:compSC}. Since $\C(H)$ is a block graph, connected components of $\C(H) \setminus V(T_A)$ are adjacent to either a critical clique of $T_A$ or a maximal clique of $T_A$ (else, there would be a diamond in $\C(H)$). 
We claim that there are $O(k^3)$ critical cliques in the connected components of $\C(H) \setminus V(T_A)$ adjacent to maximal cliques of $T_A$. Since $T_A$ is a block graph and $|V(T_A)| = O(k^2)$, there are $O(k^2)$ maximal cliques in $T_A$. Moreover, there is only one connected component of $\C(H) \setminus V(T_A)$ adjacent to each maximal clique, otherwise there would be a diamond in $\C(H)$. Take $K$ a maximal clique of $T_A$ and let $X$ be its adjacent connected component of $\C(H) \setminus V(T_A)$. Observe that $X$ has to be a union of $1$-SC-branches and their attachment points form a clique. By \cref{rule:voisinage_communBBE}, there are at most $2k+1$ $1$-SC-branches in $X$ and each one has been reduced by \cref{rule:1-SC-branchE}, thus $X$ contains at most $4k+2$ critical cliques.

Finally, we claim that there are $O(k^2)$ critical cliques in the connected components of $\C(H) \setminus V(T_A)$ adjacent to critical cliques of $T_A$. First take a critical clique of $T_A\setminus A$ and its adjacent connected components of $\C(H) \setminus V(T_A)$, recall that $|T_A\setminus A| = O(k^2)$. Observe that they form a $1$-SC-branch, reduced by \cref{rule:1-SC-branchE}, thus the adjacent connected component actually consists in only one critical clique. Next, take a critical clique $a$ of $A$, and let $C_1, \dots, C_r$ be the connected components of $\C(H) \setminus V(T_A)$ adjacent to $a$ and $Y$ the union of these connected components (see \cref{fig:ex_graphe_reduit}). Observe that each $C_i$ has to be a union of $1$-SC-branches and their attachment points form a clique $Q_i$. Let $B_1, \dots, B_l$ with attachment points $P_1, \dots , P_l$ be the $1$-SC-branches of $Y$. By \cref{rule:voisinage_communBBE}, there are at most $2k+1$ $1$-SC-branches by connected component and if $\Sigma_{i=1}^l |P_i| > 2k+1$, then $(\Sigma_{i=1}^l |P_i|) - \max_{1\leq j \leq r}\{|Q_j|\} \leq k$, implying that there are at most $3k+1$ $1$-SC-branches in $Y$. Each of these $1$-SC-branches is reduced by \cref{rule:1-SC-branchE}, hence $Y$ contains at most $6k+2$ critical cliques in total. Since $|A| = O(k)$, there is $O(k)$ such sets of $1$-SC-branches in $\C(H) \setminus V(T_A)$. Altogether, it remains that $\C(H) \setminus V(T_A)$ contains $O(k^2)$ critical cliques adjacent to critical cliques of $T_A$.

\begin{figure}[h]
    \centering
    \includegraphics[width=10cm]{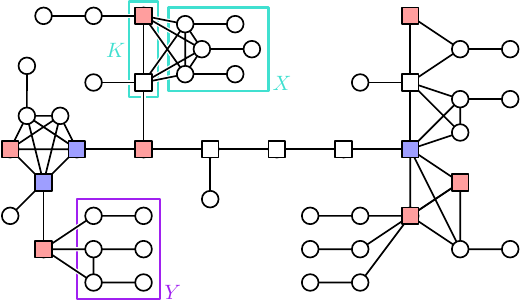}
    \caption{Illustration of the critical clique graph of a reduced instance in the proof of \cref{thm:taille_noyau_SCE}. Square nodes correspond to vertices of $T_A$, the ones filled in red are vertices of $A$, the ones filled in blue are vertices of $\sd{T_A}\setminus A$.}
    \label{fig:ex_graphe_reduit}
\end{figure}

 We conclude that $|V(\C(H))| = O(k^3)$. Each unaffected critical clique of $H$ contains $k+1$ vertices by \cref{rule:borneCC}. Moreover, each affected critical clique of $H$ contains at most $k+1$ unaffected vertices since they are clique modules of $G$ (in particular, they were contained in a critical clique of $G$ which was reduced by \cref{rule:borneCC}). Therefore, taking into accounts the at most $2k$ affected vertices, we have that $|V(G)|\leq (k+1)|V(\C(H))|+2k = O(k^4)$. Therefore, given a reduced instance $(G,k)$ of \SCE{}, if $|V(G)| > g(k)$ for a function $g(k)= O(k^4)$ (defined by this proof) we can safely conclude that $(G,k)$ is a no-instance and return a trivial no-instance of constant size, thus proving the size of our kernel. 
To finish the proof, we claim that a reduced instance can be computed in polynomial time. Indeed, \cref{lem:safecc} states that it is possible to reduce exhaustively a graph under \cref{rule:compSC,rule:borneCC} in linear time. Once this is done, it remains to apply exhaustively \cref{rule:1-SC-branchE,rule:voisinage_communBBE,rule:2-SC-branchE} which is ensured by \cref{lem:detect}.
\end{proof}


\section{Kernels for the completion and deletion variants} 
\label{sec:sccd}

In this section we present kernels for \SCC{} and \SCD{}. It is clear that \cref{rule:compSC,rule:borneCC} are also safe for these variants of the problem. We can also observe that \cref{rule:1-SC-branchE,rule:voisinage_communBBE,rule:2-SC-branchE} are safe for the completion and deletion variants since they are a special case of the edition. 
However, we can use more efficient rules to manage the $2$-SC-branches and bound their length to a constant number of critical cliques. This is due to only allowing edge addition or edge deletion. In the case of \SCC{}, one can observe that for a graph $G$, it is sufficient to consider the clean $2$-SC-branches $B$ of length $3$ which have attachment points that belong to different connected components of $G\setminus V(B^R)$ since a SC-branch remains connected in any completion of $G$. The next lemma 
follows from this observation and its proof is contained in the proof of \cref{lem:2-SC-branchE}. From this lemma we can devise an improved rule for the completion variant.

\begin{lemma} \label{lem:2-SC-branchC}
Let $G=(V,E)$ be graph and $B$ a clean $2$-SC-branch of length at least 3 with attachment points $P_1,P_2$ that are in different connected components of $G\setminus V(B^R)$. There exists an optimal completion $F$ of $G$ such that:
\begin{itemize}
    \item The set of vertices of $B$ affected by $F$ is included in $P_1 \cup N_B(P_1) \cup P_2 \cup N_B(P_2)$,
    \item In $H = G + F$, the vertices of $N_B(P_1)$ (resp. $N_B(P_2)$) are all adjacent to the same vertices of $V(G) \setminus V(B^R)$.
\end{itemize}
\end{lemma}

\begin{polyrule} \label{rule:2-SC-branchC}
Let $(G=(V,E),k)$ be an instance of \SCC{}  and $B$ a clean $2$-SC-branch of $G$ of length at least $3$ with attachment points $P_1, P_2$ that are in different connected components of $G\setminus V(B^R)$. Remove the vertices $V(B^R)$ and add two new cliques $K_1$ and $K_2$ of size respectively $min\{|N_B(P_1)|, k+1\}$ and $min\{|N_B(P_2)|, k+1\}$ with edges $(P_1\times K_1), (K_1\times K_2), (K_2 \times P_2)$.

\end{polyrule}

In the case of \SCD{}, we can observe that if the two attachment points of $B$, a $2$-SC-branch of length at least $3$, are in the same connected component of $G\setminus V(B^R)$, then either $B$ is disconnected or the attachment points are not in the same connected component of $(G - F)\setminus V(B^R)$ for any deletion $F$ of $G$. Indeed, if this was not the case, we could find an induced cycle of length at least $4$ since edge additions are not allowed. The next lemma follows from this observation and its proof is contained in the proof of \cref{lem:2-SC-branchE}. As previously we can devise an improved rule from this lemma for the deletion variant.

\begin{lemma} \label{lem:2-SC-branchD}
Let $G=(V,E)$ be graph and $B$ a clean $2$-SC-branch of length at least $3$ with attachment points $P_1,P_2$. There exists an optimal deletion $F$ of $G$ such that:
\begin{itemize}
    \item  If $P_1$ and $P_2$  are in the same connected component of $(G\triangle F)\setminus V(B^R)$, then $F$ contains an edge-cut of $B$. 
    This edge-cut is $(P_1 \times N_B(P_1))$ or $(P_2 \times N_B(P_2))$ or a min-cut of $B$.
    \item In each case, the other vertices of $B$ affected by $F$ are included in $P_1 \cup N_B(P_1) \cup P_2 \cup N_B(P_2)$,
    \item In $G - F$ the vertices of $N_B(P_1)$ (resp. $N_B(P_2)$) are all adjacent to the same vertices of $V(G) \setminus V(B^R) $.
\end{itemize}
\end{lemma}

\begin{polyrule} \label{rule:2-SC-branchD}
Let $(G=(V,E),k)$ be an instance of \SCD{} and $B$ a clean $2$-SC-branch of $G$ of length at least $3$ with attachment points $P_1, P_2$. Remove the vertices $V(B^R)$, and add the following path of $5$ cliques:  
\begin{center}
$K_{min\{|N_B(P_1)|, k+1\}} - K_{k+1} - K_1 - K_{mc(B)} - K_{min\{|N_B(P_2)|, k+1\}}$
\end{center} 
\noindent where $K_n$ is the clique of size $n$ and  $K_{min\{|N_B(P_1)|, k+1\}}$ (resp. $K_{min\{|N_B(P_2)|, k+1\}}$) is adjacent to $P_1$ (resp. $P_2$).
\end{polyrule}

\begin{lemma}
\cref{rule:2-SC-branchC,rule:2-SC-branchD} are safe and can be applied in polynomial time.
\end{lemma}

With these new rules, we can observe that in the proof of \cref{thm:taille_noyau_SCE}, the connected minimal induced subgraph $T$ that spans all the affected critical cliques contains $O(k)$ critical cliques. Indeed, the quadratic number of critical cliques in $T$ resulted from the paths in $T$ that contained a linear number of critical cliques.  However, with \cref{rule:2-SC-branchC,rule:2-SC-branchD}, these paths contain $O(1)$ critical cliques, implying that $T$ contains $O(k)$ critical cliques. With this observation, it follows that a reduced graph contains $O(k^2)$ critical cliques and thus contains $O(k^3)$ vertices.


\THMSCCD*

\section{Conclusion}
\label{sec:conclu}

We presented polynomial kernels for the edition and deletion variants of block graph edge modification problems and the three variants of strictly chordal edge modification problems. 
Our conviction is that the approach based on decompositions of the target class, combined with the ability of reducing the size of the bags of the decomposition and of limiting the number of affected bags to $O(k)$ is a promising starting point for edge modification problems, especially into subclasses of chordal graphs. The technique has been employed especially for classes that admit a tree-like decompositions with disjoint bags (e.g., 3-leaf power~\cite{BPP10}, trivially perfect graphs~\cite{DPT23,DP23} and $\mathcal{M}$-free graphs~\cite{KU14}), also for other types of tree-like decompositions with non-disjoint bags (e.g., ptolemaic graphs~\cite{CGP21}). We generalize it here to strictly chordal graphs, that have a decomposition into disjoint bags as nodes of a block graph.

The difficulty is that, at this stage, each class requires ad-hoc arguments and reduction rules, based on its specific decomposition. An ambitious goal would be to obtain a generic algorithm for edge modification problems into any class of chordal graphs, plus a finite set of obstructions, as conjectured by Bessy and Perez~\cite{BP13} for completion problems. As an intermediate step, we ask whether $4$-leaf power completion, deletion and editing problems admit a polynomial kernel. 

On another level, it would be really interesting to obtain polynomial kernelization lower bounds for edge-modification problems. To the best of our knowledge, such bounds are only known for vertex-deletion problems; more precisely, any such problem where the target class is closed under subgraphs does not admit a kernel with $O(k^{2-\epsilon})$ edges unless $coNP \subseteq NP/poly$, see~\cite{DvM14}. 

\bibliography{mybib}
\end{document}